\documentclass[10pt,conference,letterpaper]{IEEEtran}

\usepackage{times,amsmath,epsfig}
\usepackage{booktabs} 
\usepackage{amsthm}
\usepackage[british]{babel} 
\usepackage{subcaption}
\usepackage{algorithm}  
\usepackage{algpseudocode}
\usepackage[noadjust]{cite}
\usepackage{tikz}
\usetikzlibrary{positioning,shadows,arrows,trees,shapes,fit}
\usepackage{url}
\usepackage{amsfonts}
\usepackage{scrextend}
\newtheorem{theorem}{Theorem}[section]
\newtheorem{corollary}[theorem]{Corollary}
\newtheorem{definition}[theorem]{Definition}
\newtheorem{example}[theorem]{Example}
\newtheorem{lemma}[theorem]{Lemma}
\newtheorem{remark}[theorem]{Remark}

\title{Secure Range Queries for Multiple Users}

\author{\IEEEauthorblockN{Anselme Tueno}
\IEEEauthorblockA{SAP Security Research \\
Karlsruhe, Germany \\
anselme.kemgne.tueno@sap.com}
\and
\IEEEauthorblockN{Florian Kerschbaum}
\IEEEauthorblockA{University of Waterloo, Canada\\
 Waterloo, Canada\\
florian.kerschbaum@uwaterloo.ca}
}

\begin{document}

\maketitle

\begin{abstract}
Order-preserving encryption allows encrypting data, while still enabling efficient range queries on the encrypted data. Moreover, it does not require any change to the database management system, because comparison operates on ciphertexts as on plaintexts. This makes order-preserving encryption schemes very suitable for data outsourcing in cloud computing scenarios. However, all order-preserving encryption schemes are necessarily symmetric limiting the use case to one client and one server.
Imagine a scenario where a Data Owner encrypts its data before outsourcing it to the Cloud Service Provider and a Data Analyst wants to execute private range queries on this data. This scenario occurs in many cases of collaborative machine learning where data source and processor are different entities. Then either the Data Owner must reveal its encryption key or the Data Analyst must reveal the private queries. 
In this paper, we overcome this limitation by allowing the equivalent of a public-key order-preserving encryption.
We present a secure multiparty protocol that enables secure range queries for multiple users. In this scheme, the Data Analyst cooperates with the Data Owner and the Cloud Service Provider in order to order-preserving encrypt the private range queries without revealing any other information to the parties.
The basic idea of our scheme is to replace encryption with a secure, interactive protocol.
In this protocol, we combine order-preserving encryption based on binary search trees with homomorphic encryption and garbled circuits achieving security against passive adversaries with sublinear communication and computation complexity.
We apply our construction to different order-preserving encryption schemes including frequency-hi\-ding order-preserving encryption which can withstand many of the popularized attacks on order-preserving encryption.
We implemented our scheme and observed that if the database size of the Data Owner has 1 million entries it takes only about 0.3 s on average via a loopback interface (1.3 s via a LAN) to encrypt an input of the Data Analyst.
\end{abstract}


\maketitle

\section{Introduction}
\label{Sec_Introduction}
In cloud computing, companies use a network of remote servers hosted by a service provider on the Internet to store, manage, and process data, rather than a local server or a personal computer. Naively this would imply that Data Owners must give up either the security of the data or the functionality of processing the data. Therefore companies are reluctant to migrate their sensitive data to the cloud. However, different techniques, such as secure multiparty computation \cite{Yao.1982, LindellP.2002, Goldreich.2004, PinkaSSW.2009, LindellP.XX.2009, CramerDN.2015}, homomorphic encryption \cite{Gentry.2009, Paillier.1999, DamgardJurik.2001} or order-preserving encryption \cite{AgrawalKSX.2004, BoldyrevaCLO.2009, BoldyrevaCO.2011, PopaRZB.2011, PopaLZ.2013, KerschbaumS.2014, Kerschbaum.2015}, exist that enable cloud users to encrypt their data before outsourcing it to the cloud while still be able to process and search on the outsourced and encrypted data without decrypting it. Order-preserving encryption (OPE) allows encrypting data, while still enabling efficient range queries on the encrypted data. Moreover it does not require any change to the database management system, because comparison operates on ciphertexts. This makes order-preserving encryption schemes very suitable for data outsourcing in cloud computing scenarios, since it can be retrofitted to existing applications.

However, all OPE schemes are necessarily symmetric limiting the use case to one client and one server. This is due to the fact that a public-key encryption would allow a binary search on the ciphertext. Imagine a scenario where a Data Owner (DO) encrypts its data before outsourcing it to the Cloud Service Provider (CSP) and a Data Analyst (DA) wants to execute private range queries on this data. Then either the Data Owner must reveal its encryption key, since order-preserving encryption is symmetric, or the Data Analyst must reveal the private queries. 

This distinction between DO and DA occurs in many cases of collaborative data analysis, data mining and machine learning. In such scenarios, multiple parties need to jointly conduct data analysis tasks based on their private inputs. As concrete examples from the literature consider, e.g., supply chain management, collaborative forecasting, benchmarking, criminal investigation, smart metering, etc.) \cite{DuA.2001, AtallahDES.2003, AtallahBLFT.2004, Kerschbaum.2012}. Although in these scenarios plaintext information sharing would be a viable alternative, participants are reluctant to share their information with others. This reluctance is quite rational and commonly observed in practice. It is due to the fact that the implications of the information are unknown or hard to assess. For example, sharing the information could weaken their negotiation position, impact customers' market information by revealing corporate performance and strategies or impact reputation \cite{AtallahDES.2003, AtallahBLFT.2004, CatrinaK.2008}.

In this paper, we overcome the limitation of private range querying on order-preserving encrypted data by allowing the equivalent of a public-key encryption. Our idea is to replace public-key encryption with a secure, interactive protocol. Non-interactive binary search on the ciphertext is no longer feasible, since every encryption requires the participation of the Data Owner who can rate limit (i.e., control the rate of query sent by) the Data Analyst. 

Since neither the DA wants to reveal his query value nor the DO his encryption state (key), this is clearly an instance of a secure computation where two or more parties compute on their secret inputs without revealing anything but the result.  In an ideal world the DA and DO would perform a two-party secure computation for the encryption of the query value and then the DA would send the encrypted value as part of an SQL query to the CSP.  However, this two-party secure computation is necessarily linear in the encryption state (key) and hence the size of the database.  Our key insight of this paper is that we can construct an encryption with logarithmic complexity in the size of the database by involving the CSP in a three-party secure computation without sacrificing any security, since the CSP will learn the encrypted query value in any case.  One may conjecture that in this construction the encryption key of the DO may be outsourced to secure hardware in the CSP simplifying the protocol to two parties, but that would prevent the DO from rate limiting the encryption and the binary search attack would be a threat again, even if the protocol were otherwise secure.

We call our protocol \textit{oblivious order-preserving encryption}.
We implemented it and an encryption by the DA takes 0.3 seconds using the loopback interface and 1.3 seconds using a LAN in a large data center.


Our contributions are as follows:
\begin{itemize}
	\item First, we introduce a novel notion of oblivious order-preser\-ving encryption. This scheme allows a DA to execute private range queries on an order-preserving 
	encrypted database.
	\item Then, we propose an oblivious order-preserving encryption protocol based on mutable order-preserving encryption sche\-mes by Popa et al. \cite{PopaLZ.2013} 
	and Kerschbaum and Schr{\"{o}}pfer \cite{KerschbaumS.2014}.
	\item Since the schemes \cite{PopaLZ.2013, KerschbaumS.2014} are deterministic, we also consider the case where the underlying OPE scheme is the frequency-hiding 
	OPE of \cite{Kerschbaum.2015}, which is probabilistic.
	\item Finally, we implement and evaluate our scheme.
\end{itemize}
The remainder of the paper is structured as follows. We review related work in Section \ref{Sec_Related_Work} and preliminaries in Section \ref{Sec_Preliminaries} before defining correctness and security of oblivious OPE in Section \ref{Sec_Correctness_Def}. Section \ref{Sec_OOPE_Protocol} describes our scheme for the case where the underlying OPE scheme is deterministic and prove its correctness and security. In Section \ref{Sec_Int_Comparison} we discuss integer comparison and equality test with garbled circuit and how we use it in our schemes. The non-deterministic case is handled in Section \ref{Sec_OOPE_Protocol_FHOPE}. We discuss implementations details and evaluation in Section \ref{Sec_Implementation} before concluding our work in Section \ref{Sec_Conclusion}.

\section{Related Work}
\label{Sec_Related_Work}

Order-preserving encryption can be classified into stateless schemes (see Section~\ref{subsec_less}) and stateful schemes (see Section~\ref{MutableOPEDef}). Our work is concerned with stateful schemes and hence we introduce some of their algorithms in this section. However, we review stateless schemes and their security definitions first in order to distinguish them from stateful schemes.

\subsection{Stateless Order-preserving Encryption}
\label{subsec_less}
Order-preserving encryption ensures that the order relation of the ciphertexts is the same as the order of the corresponding plaintexts. This allows to efficiently search on the ciphertexts using binary search or perform range queries without decrypting the ciphertexts. The concept of order-preserving encryption was introduced in the database community by Agrawal et al.~\cite{AgrawalKSX.2004}. The cryptographic study of Agrawal et al.'s scheme was first initiated by \cite{BoldyrevaCLO.2009}, which proposed an ideal security definition IND-OCPA\footnote{IND-OCPA means indistinguishability under ordered chosen plaintext attacks and requires that OPE schemes must reveal no additional information about the plaintext values besides their order} for OPE. The authors proved that under certain implicit assumptions IND-OCPA is infeasible to achieve. Their proposed scheme was first implemented in the CryptDB tool of Popa et al.~\cite{PopaRZB.2011} and attacked by Naveed et al.~\cite{NaveedKW.2015}. In \cite{BoldyrevaCO.2011} Boldyreva et al.~further improved the security and introduced modular order-preserving encryption (MOPE). MOPE adds a secret modular offset to each plaintext value before it is encrypted. It improves the security of OPE, as it does not leak any information about plaintext location, but still does not provide ideal IND-OCPA security. Moreover Mavroforakis et al.~showed that executing range queries via MOPE in a naive way allows the adversary to learn the secret offset and so negating any potential security gains. They address this vulnerability by introducing query execution algorithms for MOPE \cite{MavroforakisCOK.2015}. However, this algorithm assumes a uniform distribution of data and has already been attacked in \cite{DurDuB16}. In a different strand of work Teranishi et al.~improve the security of stateless order-preserving encryption by randomly introducing larger gaps in the ciphertexts \cite{TerYun14}. However, they necessarily also fail at providing ideal security.

\subsection{Stateful Order-preserving Encryption} 
\label{MutableOPEDef}
Popa et al.~were the first to observe that one can avoid the impossibility result of \cite{BoldyrevaCLO.2009} by giving up certain restrictions of OPE. As result of their observations they introduced mutable OPE \cite{PopaLZ.2013}. Their first observation was that most OPE applications only require a less restrictive interface than that of encryption schemes. Their encryption scheme is therefore implemented as an interactive protocol running between a client that also owns the data to be encrypted and an honest-but-curious server that stores the data. Moreover, it is acceptable that a small number of ciphertexts of already-encrypted values change over time as new plaintexts are encrypted. With this relaxed definition their scheme was the first OPE scheme to achieve ideal security.

\textbf{Popa et al.'s scheme (mOPE$_1$) \cite{PopaLZ.2013}.} The basic idea of Popa et al.'s scheme is to have the encoded values organized at the server in a binary search tree (\textit{OPE-tree}). Specifically the server stores the state of the encryption scheme in a table (\textit{OPE-table}). The state contains ciphertexts consisting of a deterministic AES ciphertext and the order (\textit{OPE Encoding}) of the corresponding plaintext. To encrypt a new value $x$ the server reconstructs the OPE-tree from the OPE-table and traverses it. In each step of the traversal the client receives the current node $v$ of the search tree, decrypts and compares it with $x$. If $x$ is smaller (resp.~larger) then the client recursively proceeds with the left (resp.~right) child node of $v$. An edge to the left (resp.~to the right) is encoded as 0 (resp.~1). The OPE encoding of $x$ is then the path from the root of the tree to $x$ padded with $10 \ldots 0$ to the same length $l$. To ensure that the length of OPE encoding do not exceed the defined length $l$, the server must occasionally perform balancing operations. This  updates some order in the OPE table (i.e.~the OPE encoding of some already encrypted values mutate to another encoding). 

\textbf{Kerschbaum and Schr{\"{o}}pfer's scheme (mOPE$_2$) \cite{KerschbaumS.2014}.} The insertion cost of Popa et al's scheme is high, because the tree traversal must be interactive between the client and the server. To tackle this problem Kerschbaum and Schr{\"{o}}pfer proposed in \cite{KerschbaumS.2014} another ideal secure, but significantly more efficient, OPE scheme. Both schemes use binary search and are mutable, but the main difference is that in the scheme of \cite{KerschbaumS.2014} the state is not stored on the server but on the client. Moreover the client chooses a range $\{0, \ldots, M\}$ for the order. For each plaintext $x $ and the corresponding OPE encoding $y \in \{0, \ldots, M\}$ the client maintains a pair $\langle x, y\rangle$ in the state. To insert a new plaintext the client finds two pairs $\langle x_i, y_i\rangle$, $\langle x_{i+1}, y_{i+1}\rangle$ in the state such that $x_i \leq x < x_{i+1}$ and computes the OPE encoding  as follows:
\begin{itemize}
	\item if $x_i = x$ then the OPE encoding of $x$ is $y = y_i$
	\item else
		\begin{itemize}
			\item if $y_{i+1} - y_{i} = 1$ then 
				\begin{itemize}
					\item update the state (Algorithm 2 in \cite{KerschbaumS.2014})\footnote{This potentially updates all OPE encoding $y$ produced so far 
					\cite{KerschbaumS.2014}.}. 
				\end{itemize}
			\item the OPE encoding of $x$ is $y = y_i + \lceil \frac{y_{i+1} - y_{i}}{2} \rceil$.
		\end{itemize}
\end{itemize}
The encryption algorithm is keyless and the only secret information is the state which grows with the number of encryptions of distinct plaintexts. The client uses a dictionary to keep the state small and hence does not need to store a copy of the data.

\textbf{Kerschbaum's scheme (mOPE$_3$) \cite{Kerschbaum.2015}.}  Deterministic OPE sche\-mes \cite{AgrawalKSX.2004, BoldyrevaCLO.2009, PopaLZ.2013, KerschbaumS.2014, TerYun14} are vulnerable to many attacks like: frequency analysis, sorting attack, cumulative attack \cite{NaveedKW.2015, GruSek16}. To increase the security of OPE Kerschbaum first introduced in \cite{Kerschbaum.2015} a new security definition called \textit{indistinguishability under frequency-analyzing ordered chosen plaintext attack} (IND-FAOCPA) that is strictly stronger than IND-OCPA. Second he proposed a novel OPE scheme mOPE$_3$ that is secure under this new security definition. The basic idea of this scheme is to randomize ciphertexts such that no frequency information from repeated ciphertexts leaks. It borrows the ideas of \cite{KerschbaumS.2014} with a modification that re-encrypts the same plaintext with a different ciphertext. First client and server state are as by mOPE$_2$. The order ranges from $0$ to $M$ as by mOPE$_2$. The algorithm traverses the OPE-tree by going to the left or to the right depending on the comparison between the new plaintext and nodes of the tree. However, if the value being encrypted is equal to some value in the tree then the algorithm traverses the tree depending on the outcome of a random coin. Finally, if there is no more node to traverse the algorithm rebalances the tree if necessary and then computes the ciphertext similarly to $y = y_i + \lceil \frac{y_{i+1} - y_{i}}{2} \rceil$. 
 
In subsequent independent analysis~\cite{GruSek16} this encryption scheme has been shown to be significantly more secure to the attacks against order-preserving encryption (albeit not perfectly secure).

\section{Preliminaries}
\label{Sec_Preliminaries}

\subsection{Problem Statement}
\begin{figure}[h]
\centering
\includegraphics[width=0.45\textwidth]{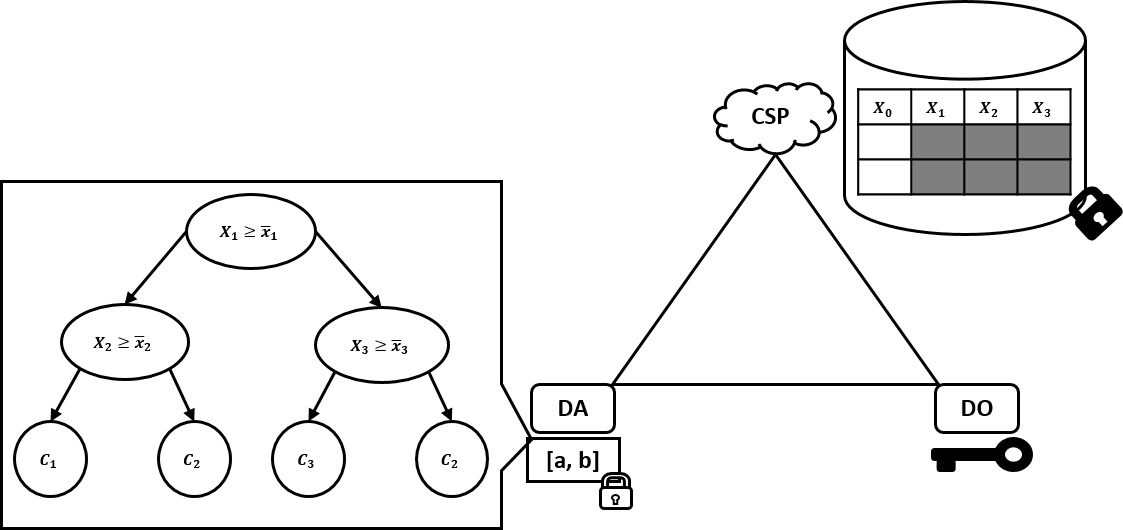}
\caption{Illustration of the Problem: \textmd{DO sends encrypted data to CSP and retains encryption keys. DA holds a private decision tree that can be represented as set of range queries. DA wants to perform data analysis on DO's encrypted data without revealing any information on the queries. DO wants to maintain privacy of the data stored at CSP.}} 
\label{Problem_Def}
\end{figure}
Our work is motivated by the following scenario. Assume a Data Owner encrypts its data with an order-preserving encryption, stores the encrypted data in a cloud database held by a Cloud Service Provider, and retains the encryption key. A Data Analyst wants to perform some analysis on the encrypted data. To this end it holds, e.g., a private machine learning model involving comparisons. In a supply chain scenario, the Data Analyst could be a supplier (manufacturer) wanting to optimize its manufacturing process based on data owned by its buyer (another supplier or distributor).

For instance, we assume the model to be a decision tree as pictured in Figure \ref{Problem_Def}, where the $\overline{x}_i$ are the thresholds and $(X_1, X_2, X_3)$ is the input vector (that maps to corresponding columns in the  Data Owner's database) to be classified. In order to use the model for classification the Data Analyst transforms the decision tree into range queries, e.g., for class $c_1$ we have the query $(X_1 < \overline{x}_1) \wedge (X_2 < \overline{x}_2)$. More precisely the Data Analyst wants to execute queries like in equations \ref{Select_Count} and \ref{Select_Column}, where we assume $X_0$ to be public.
\begin{align}
   \textsc{Select} ~ \textsc{Count}(*) ~ \textsc{Where} ~ X_1 < \overline{x}_1 ~ \textsc{And} ~ X_2 < \overline{x}_2 \label{Select_Count}\\
	\textsc{Select} ~ X_0 ~ \textsc{Where} ~ X_1 < \overline{x}_1 ~ \textsc{And} ~ X_2 < \overline{x}_2 \label{Select_Column}
\end{align}

However, as the database is encrypted (i.e. columns $X_1$, $X_2$ and $X_3$ are OPE encrypted) the Data Analyst needs ciphertexts of the thresholds $\overline{x}_i$. 
In \cite{TaigelTP.2018}, Taigel et. al. describes an approach that combines decision tree classification and OPE to enable privacy-preserving forecasting of maintenance demand based on distributed condition data. They consider the problem of a Maintenance, Repair, and Overhaul (MRO) provider from the aerospace industry that provides maintenance services to their customers' (e.g., commercial airlines or air forces) jet engines. The customers as Data Owners consider the real condition data of their airliners as very sensitive and therefore this data is stored encrypted in a cloud database using OPE. The MRO provider as Data Analyst holds a decision tree that can predict the probability of maintenance, repair, and overhaul of spare parts. However, the classification of an individual spare part is not necessary, but only aggregated numbers such as returned by Equation \ref{Select_Count}. The aggregated numbers then allow the MRO to compute the forecast without violating the privacy of the real condition data. In \cite{TaigelTP.2018}, they provide privacy only for the customer, while we want to allow privacy for both customer and MRO provider.

Order-preserving encryption is necessarily symmetric, thus only the Data Owner can encrypt and decrypt the data stored on the cloud server. If the Data Analyst needs to obtain a ciphertext for a query, it can just send the plaintext threshold to the Data Owner. However, if the model contains intellectual property which the Data Analyst wants to remain protected, then this free sharing of information is no longer possible. Our goal is to allow the data analysis to be performed efficiently without revealing any sensitive information in the query and without revealing the key to the order-preserving encryption.

We now review a few building blocks used in our construction of oblivious order-preserving encryption.

\subsection{Secure Multiparty Computation}
Secure multiparty computation (SMC) is a cryptographic technique that allows several parties to compute a function on their private inputs without revealing any information other than the function's output. A classical example in the literature is the so called \textit{Yao's millionaire's problem} introduced in \cite{Yao.1982}. Two millionaires are interested in knowing which of them is richer without revealing their actual wealth. Formally we have a set of $n$ parties $P_1, \ldots , P_n$, each with its own private input $x_1, \ldots, x_n$ and they want to compute the function $y=f(x_1, \ldots, x_n)$\footnote{The output of each party can also be private, in this case $y = (y_1, \ldots, y_n)$.} without disclosing their private inputs.

Security of SMC protocols is often defined by comparison to an \textit{ideal model}. In that model parties privately send their input to a trusted third party (TTP). Then the TTP computes the outcome of the function on their behalf, sends the corresponding result to each party and forgets about the private inputs. In the \textit{real model} parties emulate the ideal model by executing a cryptographic protocol to perform the computation.
At the end only the result should be revealed and nothing else. A SMC protocol is then said to be secure if the adversary can learn only the result of the computation and data that can be deduced from this result and known inputs \cite{Goldreich.2004, CramerDN.2015, Frikken.2010}.

An important issue to consider when defining the security of SMC is the adversary's power. There exists many security models, but the \textit{semi-honest} and the \textit{malicious} adversary model are the most popular \cite{Goldreich.2004, CramerDN.2015}. In the semi-honest (a.k.a \textit{honest-but-curious}) model parties behave \textit{passively} and follow the protocol specification. However, the adversary can obtain the internal state of corrupted parties and uses this to learn more information. 
In contrast, a malicious adversary is \textit{active} and instructs corrupted parties to deviate from the protocol specification. 

\subsection{Yao's Garbled Circuit}

Yao's initial protocol for secure two-party computation uses a technique called Garbled Circuits (GC). A GC can be used to execute a function over symmetrically encrypted inputs. In this section, we recall the idea of GC protocol and refer to \cite{Yao.1982, LindellP.2009, PinkaSSW.2009, LindellP.XX.2009, EjgenbergFLL.2012, BellareHR.2012} for more technical description of circuit garbling and its implementation. Let $f$ be a function over two inputs $x$ and $y$, then a garbling scheme consists of a five-tuple of algorithms $\mathcal{G} = (Gb, En, De, Ev, ev)$. The original function $f$ is encoded as circuit that the function $ev(f,\cdot,\cdot):\{0,1\}^n \times \{0,1\}^n \rightarrow \{0,1\}^m$ can evaluate. On input $f$ and security parameter $k\in\mathbb{N}$, algorithm $Gb$ returns a triple of strings $(F,e,d) \leftarrow Gb(1^k ,f)$. The string $F$ describes a garbled function, $Ev(F, \cdot, \cdot)$, that maps each pair of garbled inputs $(X, Y)$ to a garbled output $Z = Ev(F,X, Y)$. The string $e$ describes an encoding function, $En(e, \cdot)$, that maps initial inputs $x, y \in \{0,1\}^n$ to garbled inputs $X = En(e,x)$,  $Y = En(e,y)$. The String $d$ describes a decoding function, $De(d, \cdot)$, that maps a garbled output $Z$ to a final output $z = De(d,Z)$. The garbling scheme is correct if $De(d, Ev(F, En(e, x), En(e, y))) = ev(f, x, y)$ \cite{BellareHR.2012}.

A GC protocol is a 2-party protocol consisting of a generator (Gen) and an evaluator (Eva) with input $x$ and $y$ respectively. On input $f$ and $k$, Gen runs $(F,e,d) \leftarrow Gb(1^k ,f)$ and parses $e$ as $(X_1^0, X_1^1, \cdots, X_n^0, X_n^1, Y_1^0, Y_1^1, \cdots, Y_n^0, Y_n^1)$. Then she sends $F, d$ and $X = (X_1^{x_1}, \cdots, X_n^{x_n}) \leftarrow En(e, X)$ to Eva, where $x_i$ represents the $i$-th bit of $x$. Now the parties execute an oblivious transfer protocol with Eva having selection string $y$ and Gen having inputs ($Y_1^0, Y_1^1, \cdots, Y_n^0, Y_n^1)$. As a result, Eva obtains $Y = (Y_1^{y_1}, \cdots, Y_n^{y_n})$ and Gen learns nothing. Finally, Eva evaluates and outputs $z = De(d,Ev(F,X, Y))$.

\subsection{Homomorphic Encryption}
A homomorphic encryption scheme is an encryption scheme that allows computations on ciphertexts by generating an encrypted result whose decryption matches the result of operations on the corresponding plaintexts. With fully homomorphic encryption schemes \cite{Gentry.2009} one can compute any efficiently computable function. However, with the current state of the art, their computational overhead is still too high for practical applications. Efficient alternatives are additive homomorphic encryption schemes, e.g.: Paillier \cite{Paillier.1999, DamgardJurik.2001}. They allow specific arithmetic operations on plaintexts, by applying an efficient operation on the ciphertexts. Let $E(x)$ denote the probabilistic encryption of a plaintext $x$. Then the following addition property holds $E(x)E(y) = E(x + y)$, i.e., by multiplying two ciphertexts one obtains a ciphertext of the sum. In our protocol we will use the public-key encryption scheme of Paillier \cite{Paillier.1999}.

\subsection{Overview of Our Construction}
In theory generic SMC allows to compute any efficiently computable function. However, any generic SMC is at least linear in the input size, which in this case is the number of encrypted values in the database. The idea of our solution is to exploit the inherent, i.e. implied by input and output, leakage of Popa et.~al.'s OPE scheme making our oblivious OPE sublinear in the database size.
Furthermore, we exploit the advantage of (homomorphic) encryption allowing a unique, persistent OPE state stored at the CSP while being able to generate secure inputs for the SMC protocol and the advantage of garbled circuits allowing efficient, yet provably secure comparison.
Our oblivious order-preserving encryption is therefore a mixed-technique, secure multi-party computation protocol between the Data Owner, the Data Analyst and the Cloud Service Provider in the semi-honest model.

In detail, our protocol proceeds as follows:
The Data Owner outsources its OPE state to the cloud-service provider.
As already described, the state consists of an OPE-table of ciphertext, order pairs.
However, in oblivious order-preserving encryption the ciphertext is created using an additively homomorphic public-key encryption scheme instead of standard symmetric encryption.
When the DA traverses the state in order to encrypt a query plaintext, the CSP creates secret shares\footnote{Secret shares are random values that add up to the plaintext.} using the homomorphic property.
One secret share is sent to the DA and one to the DO.
The DA and DO then engage in a secure two-party computation using Yao's Garbled Circuits in order to compare the reconstruction of the secret shares (done in the garbled circuit) to the query plaintext of the DA.
The result of this comparison is again secret shared between DA and DO, i.e., neither will know whether the query plaintext is above or below the current node in the traversal.
Both parties -- DA and DO -- send their secret shares of the comparison result to the CSP which then can determine the next node in the traversal.
These steps continue until the query plaintext has been sorted into the OPE-table and the CSP has an order-preserving encoding that can be sent to the DA.
A significant complication arises from this order-preserving encoding, since it must not reveal the result of the comparison protocols to the DA (although it may be correlated to the results).
In the next two sections we provide the detailed, step-by-step formalization of the construction.

\section{Correctness and Security Definitions}
\label{Sec_Correctness_Def}
In this section we first present the system architecture and then define correctness and security of oblivious order-preserving encryption. 
\subsection{System Architecture}
Our oblivious OPE (OOPE) protocol $\Pi_{OOPE}$ extends the two-party protocols by Popa et al.'s  \cite{PopaLZ.2013} and by Kerschbaum and Schr{\"{o}}pfer \cite{KerschbaumS.2014} to a three-party protocol.
 
The first party, a.k.a Data Owner (DO), encrypts its data with an order-preserving encryption as described in Section \ref{MutableOPEDef} and stores the encrypted data in a cloud database hosted by the second party, Cloud Service Provider (CSP). The third party, Data Analyst (DA), needs to execute analytic range queries, e.g. how many values are in a given range, on the Data Owner's encrypted data. However, the DO's data is encrypted with a symmetric key OPE and the DA's queries contain sensitive information. Therefore, the DA interacts with the DO and the CSP to order-preserving encrypt the sensitive queries values without learning anything else or revealing any information on the sensitive queries values.

\subsection{Definitions}
Let $\mathbb{D} = \{x_1, \ldots, x_n\}$ be the finite data set of the DO, and $h = \log_2 n$. Let $[\![x]\!]$ denote the ciphertext of $x$ under Paillier's scheme with public key $pk$ and corresponding private key $sk$ that only the DO knows. Let $\preceq$ be the order relation on $[\![\mathbb{D}]\!] = \{[\![x_1]\!], \ldots, [\![x_n]\!]\} $ defined as: $[\![x_1]\!] \preceq [\![x_2]\!]$ if and only if $x_1 \leq x_2$. The relations $\succeq, \prec, \succ$ are defined the same way with $\geq, <, >$ respectively. Let $\mathbb{P} = \{0, \ldots, 2^l-1\}$ (e.g. $l = 32$) and $\mathbb{O} = \{0, \ldots, M\}$ ($M$ positive integer) be plaintext and order\footnote{We will use \textit{order} and \textit{OPE encoding} interchangeably.} range resp., i.e.: $\mathbb{D} \subseteq \mathbb{P}$.

We begin by defining order-preserving encryption as used in this paper.

\begin{definition}\label{DefOPE}\normalfont
Let $\lambda$ be the security parameter of the public-key scheme of Paillier. An \textit{order-preserving encryption (OPE)} consists of the three following algorithms:
\begin{itemize}
	\item $(pk, sk) \leftarrow    \textsc{KeyGen}(\lambda)$: Generates a public key $pk$ and a private key $sk$ according to $\lambda$,
	\item $S', \left\langle [\![x]\!], y\right\rangle  \leftarrow  \textsc{Encrypt}\textsubscript{\tiny OPE}(S, x, pk)$: For a plain $x \in \mathbb{P}$ computes ciphertext $\left\langle [\![x]\!], y\right\rangle$ and updates the state $S$ to $S'$, where $[\![x]\!] \leftarrow  \textsc{Paillier}.\textsc{Encrypt}(x, pk)$ is a Paillier ciphertext, $y \leftarrow $ mOPE$_2.\textsc{Encrypt}(S, x)$ (resp.  $y \leftarrow $mOPE$_3.\textsc{Encrypt}(S, x)$) is the \textit{order} of $x$ in the deterministic (resp. non-deterministic) case with $y \in \mathbb{O}$.
	\item $x \leftarrow \textsc{Decrypt}\textsubscript{\tiny OPE}(\left\langle [\![x]\!], y\right\rangle, sk)$: Computes the plaintext $x \leftarrow \textsc{Paillier}.\textsc{Decrypt}([\![x]\!], sk)$ 
	of the ciphertext $\left\langle [\![x]\!], y\right\rangle$.
\end{itemize} 
The encryption scheme is \textit{correct} if: $$\textsc{Decrypt}\textsubscript{\tiny OPE}(\textsc{Encrypt}\textsubscript{\tiny OPE}(S, x, pk), sk) = x$$ for any valid state $S$ and $x$. It is \textit{order-preserving} if the order is preserved, i.e. $y_i < y_j \Rightarrow x_i \leq x_j$ for any $i$ and $j$.
\end{definition}
For a data set $\mathbb{D}$ the encryption scheme generates an ordered set of ciphertexts. We formalize it with the following definition.  

\begin{definition}\label{DefOPETable}\normalfont
Let $j_1, j_2, \ldots$ be the ordering of $\mathbb{D}$ (i.e $x_{j_1} \leq x_{j_2} \leq \ldots$) then the OPE scheme generates an \textit{OPE-Table} which is an ordered set $\mathbb{T} = \left\langle [\![ x_{j_1}]\!], y_{j_1} \right\rangle, \left\langle [\![ x_{j_2}]\!], y_{j_2} \right\rangle, \ldots$, where $y_{j_k} \in \mathbb{O}$ is the order of $x_{j_k}$.
\end{definition}
The OPE-table is sent to the server and used to generate the following search tree during the oblivious order-preserving encryption protocol.

\begin{definition}\label{DefOPETree}\normalfont
An \textit{OPE-tree} is a tree $\mathcal{T} = (r, \mathcal{L}, \mathcal{R})$, where $r=[\![x]\!]$ for some $x$, $\mathcal{L}$ and $\mathcal{R}$ are OPE-trees such that:
If $r'$ is a node in the left subtree $\mathcal{L}$ then $r \succeq r'$ and if $r''$ is a node in the right subtree $\mathcal{R}$, then $r \preceq r''$.
\end{definition}

\begin{definition}\label{DefClt_SvrState}\normalfont
For a data set $\mathbb{D}, \textsc{Encrypt}\textsubscript{\tiny OPE}$ generates the \textit{Data Owner state}, the set of all $\left\langle x_i, y_i \right\rangle$ such that $x_i \in \mathbb{D}$ and $y_i$ is the order of $x_i$. The \textit{server state} is the pair $\mathbb{S} = \left\langle \mathcal{T}, \mathbb{T}\right\rangle$ consisting of the OPE-tree $\mathcal{T}$ and the OPE-table $\mathbb{T}$.
\end{definition}

\begin{remark}\label{Remark_OPE_Def}\normalfont
The reason of using mOPE$_2.\textsc{Encrypt}$ in Definition \ref{DefOPE} instead of mOPE$_1.\textsc{Encrypt}$ is that the DA will receive the order part of the ciphertext after the OOPE protocol. However, by mOPE$_1.\textsc{Encrypt}$ the binary representation of this order always reveals the corresponding path in the tree, allowing the DA to infer more information from the protocol than required. In contrast, mOPE$_2$ allows the DO to choose not just the length of the OPE encoding, but also the order range like $0, \ldots, M$. If $\log_2 M$ is larger than the needed length of the OPE encoding and $M$ is not a power of two, then for a ciphertext $\left\langle [\![x]\!], y\right\rangle$ $y$ does not reveal the position of $[\![x]\!]$ in the tree. In Figure \ref{lblOPETree} for instance, when applying mOPE$_1$ with $h=3$, the order of 25 (i.e. 011 = 3) reveals the corresponding path in the tree. However, with mOPE$_2$ and $M = 28$, 25 has order 11 = 1011.
\end{remark}

\begin{example}\label{ExampleState}\normalfont
Assume $\mathbb{D} = \{10, 20, 25, 32, 69\}$ is the data set, $M = 28$ and the insertion order is $32, 20, 25, 69, 10$. Then the ciphertexts after executing algorithm \textsc{Encrypt}\textsubscript{\tiny OPE} are $\left\langle [\![32]\!], 14\right\rangle$, $\left\langle [\![20]\!], 7\right\rangle$, $\left\langle [\![25]\!], 11\right\rangle$, $\left\langle [\![69]\!], 21\right\rangle$, $\left\langle [\![10]\!], 4\right\rangle$. The OPE-tree, the OPE-table and the DO state are depicted in Figure \ref{fig:lblOOPEInit}. 
\end{example}
\tikzset{
  treenode/.style = {align=center, inner sep=0pt, text centered,
    font=\sffamily},
  arn_x/.style = {treenode, rectangle, draw=black,
    minimum width=5mm, minimum height=5mm}
}
\newcommand{\OPETreeExample}{
\begin{tikzpicture}[->,level/.style={sibling distance = 3cm,
  level distance = 1cm}] 
\centering
\node [arn_x] {$ [\![32]\!]$}
    child{ node [arn_x] {$ [\![20]\!]$}
            child{ node [arn_x] {$[\![10]\!]$ } 
						edge from parent 
						node[above left = -1.5mm] {$0$}
            }
            child{ node [arn_x] {$[\![25]\!]$}
						edge from parent 
						node[above right = -1.5mm] {$1$}
            }   
					  edge from parent 
			      node[above left = -1.5mm] {$0$}
    }
    child{ node [arn_x] {$[\![69]\!]$}
		edge from parent 
		node[above right = -1.5mm] {$1$}
		}
; 
\end{tikzpicture}
}
\newcommand{\OPETableExample}{
	\centering
  \begin{tabular}{ c | c }
    \hline
    $[\![x]\!]$  & $y$ \\ \hline \hline
    $[\![10]\!]$ & $4$ \\ \hline
    $[\![20]\!]$ & $7$ \\ \hline
		$[\![25]\!]$ & $11$ \\ \hline
		$[\![32]\!]$ & $14$ \\ \hline
		$[\![69]\!]$ & $21$ \\ 
    \hline
  \end{tabular}
}
\newcommand{\OPEClienState}{
	\centering
  \begin{tabular}{ c }
    \hline
    $\left\langle x, y \right\rangle$ \\ \hline \hline
    $\left\langle 32, 14 \right\rangle$ \\ \hline
    $\left\langle 20, 7 \right\rangle$ \\ \hline
		$\left\langle 25, 11 \right\rangle$ \\ \hline
		$\left\langle 69, 21\right\rangle$ \\ \hline
		$\left\langle 10, 4 \right\rangle$ \\ 
    \hline
  \end{tabular}
}
\begin{figure}[t]
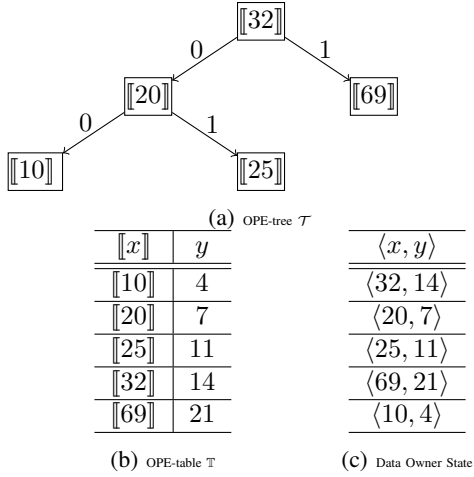

    \centering
		\begin{subfigure}[b]{0.37\textwidth}
				\OPETreeExample
				\caption{\tiny{OPE-tree $\mathcal{T}$}} 
				\label{lblOPETree}
		\end{subfigure}
		~
		\centering
		\begin{subfigure}[b]{0.40\linewidth}
				\OPETableExample
				\caption{\tiny{OPE-table $\mathbb{T}$}} 
				\label{lblOPETable}
		\end{subfigure}
		~
		\centering
		\begin{subfigure}[b]{0.25\linewidth}
				\OPEClienState
				\caption{\tiny{Data Owner State\tiny}} 
				\label{lblOPEClientState}
		\end{subfigure}
		\caption{Example initialization} 
		\label{fig:lblOOPEInit}
\end{figure}
\begin{definition}[Inputs/Outputs]\label{DefInOut}\normalfont \
The \textit{Inputs} and \textit{Outputs} of the oblivious OPE functionality are defined as follow: 
\begin{itemize}
\item Inputs
		\begin{itemize}
			\item CSP: server's state $\mathbb{S} = \left\langle \mathcal{T}, \mathbb{T} \right\rangle$
			\item DA : $\overline{x} \in \mathbb{P}$, the input to be encrypted
			\item DO : $sk$, DO's private key under Paillier 
		\end{itemize}
\item Outputs
		\begin{itemize}
			\item CSP: $\left\langle [\![\overline{x}]\!], \overline{y} \right\rangle$, the ciphertext to update the state
			\item DA : $\overline{y} \in \mathbb{O}$, the order of $\overline{x}$
			\item DO : $\emptyset$, no output.
		\end{itemize}
\end{itemize}
\end{definition}
\begin{definition}[Correctness]\label{DefCorrectness}\normalfont
Let $\mathbb{D}$ be the data set and the inputs of the protocol be defined as above. At the end of the protocol the Data Analyst obtains for its input $\overline{x}$ the output $\overline{y}$ such that $\overline{y}$ is the order of $\overline{x}$ in $\mathbb{D} \cup \{ \overline{x} \}$. The Cloud Provider obtains 
$\left\langle [\![\overline{x}]\!], \overline{y} \right\rangle$ that is added to the OPE-table. The Data Owner obtains nothing:
$$\textsc{OOPE}(\mathbb{S}, \overline{x}, sk) = (\left\langle [\![\overline{x}]\!], \overline{y} \right\rangle, \overline{y}, \emptyset)$$
$$\textsc{Decrypt}\textsubscript{\tiny OPE}(\left\langle [\![\overline{x}]\!], \overline{y} \right\rangle, sk) = \overline{x}$$
$$\left\langle [\![x_1]\!], y_1 \right\rangle, \left\langle [\![x_2]\!], y_2 \right\rangle \in \mathbb{T} ~ \wedge ~ y_1 < \overline{y} < y_2 \Rightarrow  [\![x_1]\!] \preceq [\![\overline{x}]\!] \preceq [\![x_2]\!] $$
\end{definition}
\begin{remark}\label{Rem_OOPE_Correctness}\normalfont
Updating the server state, i.e.: allowing the server to learn $[\![\overline{x}]\!]$, is only necessary if the DA wants to encrypt several values, as the encryption depends on the state.\\ 
\end{remark}

We say that two distributions $\mathcal{D}_1$ and $\mathcal{D}_2$ are computationally indistinguishable (denoted $\mathcal{D}_1 \stackrel{c}{\equiv} \mathcal{D}_1$) if no probabilistic polynomial time algorithm can distinguish them except with negligible probability. In SMC protocols the \textit{view} of a party consists of its input and the sequence of messages that it has received during the protocol execution \cite{Goldreich.2004}. The protocol is said secure if for each party, one can construct a simulator that given only the input and the output can generate a distribution that is computationally indistinguishable to the party's view.
\begin{definition}[Semi-honest Security]\label{DefSecurity}\normalfont
Let $\mathbb{D}$ be the data set with cardinality $n$ and the inputs and outputs be as previously defined. Then a protocol $\Pi$ \textit{securely} implements the functionality OOPE in the \textit{semi-honest model with honest majority} if the following conditions hold:
	\begin{itemize}
		\item there exists a probabilistic polynomial time algorithm $S_{DO}$ that simulates the DO's view $view_{DO}^{\Pi}$ of the protocol given $n$ and the private key $sk$ only,
		\item there exists a probabilistic polynomial time algorithm $S_{DA}$ that simulates the DA's view $view_{DA}^{\Pi}$ of the protocol given $n$, the input $\overline{x}$ and the output $\overline{y}$ only,
		\item there exists a probabilistic polynomial time algorithm $S_{CSP}$ that simulates the CSP's view $view_{CSP}^{\Pi}$ of the protocol given access to the server state $\mathbb{S}$ and the output $\left\langle [\![\overline{x}]\!], \overline{y} \right\rangle$ only.
	\end{itemize}
Formally:
	\begin{eqnarray}
		S_{DO}(n, sk, \emptyset) & \stackrel{c}{\equiv} & view_{DO}^{\Pi}(\mathbb{S}, \overline{x}, sk) \label{eq:OOPE_Security_Eqn_DO}\\
		S_{DA}(n, \overline{x}, \overline{y}) & \stackrel{c}{\equiv} & view_{DA}^{\Pi}(\mathbb{S}, \overline{x}, sk) \label{eq:OOPE_Security_Eqn_DA}\\
		S_{CSP}(\mathbb{S}, \left\langle [\![\overline{x}]\!], \overline{y} \right\rangle) & \stackrel{c}{\equiv} & view_{CSP}^{\Pi}(\mathbb{S}, \overline{x}, sk) \label{eq:OOPE_Security_Eqn_CSP} 
	\end{eqnarray}
\end{definition}
\section{Protocol for Oblivious OPE} 
\label{Sec_OOPE_Protocol}
In this section we present our scheme $\Pi$\textsubscript{\tiny OOPE} that consists of an initialization step and a computation step. The initialization step generates the server state and is run completely by the Data Owner. The server state and the ciphertexts are sent to the CSP afterward. 

\subsection{Initialization}
\label{Initialization}
Let $\mathbb{D} = \left\{x_1, \ldots, x_n \right\}$ be the unordered DO's dataset and $h = \log_2 n$. The DO chooses a range $0, \ldots, M$ such that $\log_2 M > h$ (Remark \ref{Remark_OPE_Def}), runs \textsc{Encrypt}\textsubscript{\tiny OPE} from Definition \ref{DefOPE} and sends the generated OPE-table to the CSP.

\subsection{Algorithms}
In our oblivious OPE protocol the CSP traverses the OPE-tree (Figure \ref{OOPE_overview}). In each step it chooses the next node depending on a previous oblivious comparison step that involve the DA and the DO. If the comparison returns equality or the CSP reaches a null node then it computes the ciphertext based on the position of the current node in the OPE-table. \\
In the following we present our main protocol that repeatedly makes calls to a sub-protocol (Protocol \ref{Comparison_Protocol}). Both protocols run between the three parties. During the protocol's execution the CSP runs Algorithm \ref{TreeTraversal_Algorithm} to traverse the tree and Algorithm \ref{Encryption_Algorithm} to compute the order. 

\textbf{Our OOPE Protocol.} As said before the protocol (Protocol \ref{OOPE_Protocol}) is executed between the three parties. First the CSP retrieves the root of the tree and sets it as current node. Then the protocol loops $h ~(= \log_2 n)$ times. In each step of the loop the CSP increments the counter and the parties run an oblivious comparison protocol (Protocol \ref{Comparison_Protocol}) whose result enables the CSP to traverse the tree (Algorithm \ref{TreeTraversal_Algorithm}). If the inputs are equal or the next node is empty then the traversal stops. However, the CSP uses the current node as input to
the next comparison until the counter reaches the value $h$. After the loop the result is either the order of the current node in case 
of equality or it is computed by the CSP using Algorithm \ref{Encryption_Algorithm}. In the last step, the DA computes $[\![\overline{x}]\!]$ using DO's public key $pk$ and sends it to the CSP as argued in Remark \ref{Rem_OOPE_Correctness}. Alternatively, the DA could generate an unique identifier (UID) for each element that is being inserted and send this UID instead. So if the corresponding node is later involved in a comparison step, the result is computed by the DA alone.

\begin{figure}[t]
\centering
\includegraphics[width=0.40\textwidth]{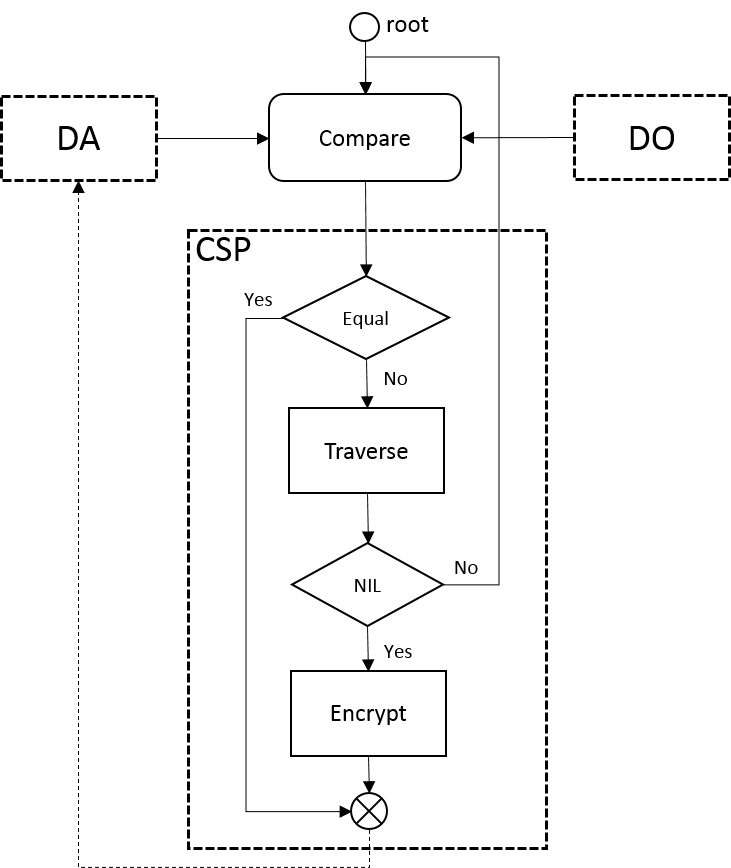}
\caption{Overview of the protocol} 
\label{OOPE_overview}
\end{figure}

\begin{algorithm}[ht]
\floatname{algorithm}{Protocol}
\caption{Oblivious OPE Protocol}
\begin{flushleft}
\textbf{Input ($In_{CSP}, In_{DA}, In_{DO}$)}: $(\mathbb{S}, \overline{x}, sk)$ \\
\textbf{Output ($Out_{CSP}, Out_{DA}, Out_{DO}$)}: $(\left\langle [\![\overline{x}]\!], \overline{y} \right\rangle, \overline{y}, \emptyset)$ \\
\textbf{Functionality }: \textsc{OOPE}($\mathbb{S}, \overline{x}, sk$)
\end{flushleft}

\hrulefill
\begin{algorithmic}[1]

\State CSP : \textbf{retrieve} root $[\![x_{root}]\!]$ of $\mathcal{T}$ 
\State CSP : \textbf{let} $[\![x]\!] \gets [\![x_{root}]\!]$
\State CSP : \textbf{let} $count \gets 0$
\Repeat
\State $(\left\langle b_e, b_g \right\rangle, \emptyset, \emptyset)\gets$ \textsc{Compare}($[\![x]\!], \overline{x}, sk$)

\State CSP : \textbf{if} $b_e \neq 0$ \textbf{then} \label{Compare_Returned_Inequality} \Comment{meaning $\overline{x} \neq x$}
\State CSP : ~~~~$[\![x_{next}]\!] \gets$ \textsc{Traverse}($b_g, [\![x]\!]$)
\State CSP : ~~~~\textbf{if} $[\![x_{next}]\!] \neq $ \textbf{NIL} \textbf{then} \label{Traverse_Returned_Not_Null_Node}
\State CSP : ~~~~~~~~\textbf{let} $[\![x]\!] \gets [\![x_{next}]\!]$
\State CSP : ~~~~\textbf{end if}
\State CSP : \textbf{end if}

\State CSP : \textbf{let} $count \gets count + 1$ 
\Until {$count = h$ }

\State CSP : \textbf{if} $b_e = 0$ \textbf{then} \label{Traverse_Stops_With_Equality} \Comment{meaning $\overline{x} = x$}
\State CSP : ~~~~\textbf{retrieve} $\left\langle [\![x]\!], y \right\rangle$ \textbf{and let} $\overline{y} \gets y$ \label{Encrypt_Trivvial}
\State CSP : \textbf{else}
\State CSP : ~~~~$\overline{y} \gets$ \textsc{Encrypt}($b_g, [\![x]\!])$
\State CSP : \textbf{end if}

\State CSP $\rightarrow$ DA: \textbf{send} $\overline{y}$ \label{OOPE_CSP_DA_Order} 
\State DA $\rightarrow$ CSP: \textbf{send} $[\![\overline{x}]\!]$ \label{OOPE_DA_CSP_Cipher}

\end{algorithmic}
\label{OOPE_Protocol}
\end{algorithm}

\textbf{Oblivious Comparison Protocol.} The oblivious comparison (Protocol \ref{Comparison_Protocol}) is a protocol between the three parties as well, with input $([\![x]\!], \overline{x}, sk)$ for the CSP, the DA and the DO respectively. 
First the CSP randomizes its input, with a random integer $r \in \{0, \ldots, 2^{l+k}\}$\footnote{Where $k$ is the security parameter that determines the statistical leakage, e.g. $k = 32$ \cite{DamgardT.2008}.}, to $[\![x+r]\!] \leftarrow [\![x]\!] \cdot [\![r]\!]$, by first computing $[\![r]\!]$ with DO's public key, such that the DO will not be able to identify the position in the tree, and it sends $[\![x+r]\!]$  to the DO and $r$ to the DA. Then the DO with input $(b_o, b'_o, x+r)$ and the DA with input $(b_a, b'_a, \overline{x}+r)$ engage in a garbled circuit protocol for comparison as described in Section \ref{Sec_Int_Comparison}. For simplicity, the garbled circuit is implemented in Protocol \ref{Comparison_Protocol} as ideal functionality. In reality the DO generates the garbled circuit and the DA evaluates it. The DA and the DO receive $(b_e\oplus b_a\oplus b_o, b_g\oplus b'_a \oplus b'_o)$ as output of this computation and resp. send $(b_a, b'_a, b_e\oplus b_o, b_g\oplus b'_o)$ and $(b_o, b'_o, b_e\oplus b_a, b_g\oplus b'_a)$ to the CSP. Finally the CSP evaluates Equation \ref{Eqn_GC_Result_Check} and outputs $\left\langle b_e, b_g \right\rangle$. This will be used to traverse the OPE-tree.
\begin{equation}
  \begin{cases}
		b_e &= b_e\oplus b_o \oplus b_o = b_e\oplus b_a \oplus b_a \\
		b_g &= b_g\oplus b'_o \oplus b'_o = b_g\oplus b'_a\oplus b'_a
	\end{cases}
	\label{Eqn_GC_Result_Check}
\end{equation}

\begin{algorithm}[ht]
\floatname{algorithm}{Protocol}
\caption{Oblivious Comparison Protocol}
\begin{flushleft}
\textbf{Input ($In_{CSP}, In_{DA}, In_{DO}$)}: $([\![x]\!], \overline{x}, sk)$ \\
\textbf{Output ($Out_{CSP}, Out_{DA}, Out_{DO}$)}: $(\left\langle b_e, b_g \right\rangle, \emptyset, \emptyset)$ \\
\textbf{Functionality }: \textsc{Compare}($[\![x]\!], \overline{x}, sk$)
\end{flushleft}

\hrulefill
\begin{algorithmic}[1]
\State CSP: choose an $(l+k)$-bits random $r$ and compute $[\![x+r]\!]$
\State CSP $\rightarrow$ DO: send $[\![x+r]\!]$
\State CSP $\rightarrow$ DA: send $r$
\State DO: decrypt $[\![x+r]\!]$ and choose masking bits $b_o, b'_o$ \label{Compare_CSP_To_DO}
\State DA: compute $\overline{x}+r$ and choose masking bits $b_a, b'_a$
\State DO $\rightarrow$ GC: send $(b_o, b'_o, x+r)$
\State DA $\rightarrow$ GC: send $(b_a, b'_a, \overline{x}+r)$ 
\State GC $\leftrightarrow$ DA: send $(b_e\oplus b_a\oplus b_o, b_g\oplus b'_a \oplus b'_o)$
\State GC $\leftrightarrow$ DO: send $(b_e\oplus b_a\oplus b_o, b_g\oplus b'_a \oplus b'_o)$
\State DA $\rightarrow$ CSP: send $(b_a, b'_a, b_e\oplus b_o, b_g\oplus b'_o)$ 
\State DO $\rightarrow$ CSP: send $(b_o, b'_o, b_e\oplus b_a, b_g\oplus b'_a)$
\State CSP: compute $b_e = b_e\oplus b_o \oplus b_o = b_e\oplus b_a \oplus b_a$
\State CSP: compute $b_g = b_g\oplus b'_o \oplus b'_o = b_g\oplus b'_a\oplus b'_a$ 
\State CSP: output $\left\langle b_e, b_g \right\rangle$
\end{algorithmic}
\label{Comparison_Protocol}
\end{algorithm}

\textbf{Tree Traversal Algorithm.} The tree traversal (Algorithm \ref{TreeTraversal_Algorithm}) runs only at the CSP. Depending on the output of the oblivious comparison the CSP either goes to the left (line \ref{Traverse_To_Left}) or to the the right (line \ref{Traverse_To_Right}). 
If the comparison step returns equality there is no need to traverse the current node and the protocol returns the corresponding ciphertext. 
\begin{algorithm}
\caption{Tree Traversal}
\begin{flushleft}
\textbf{Input:} $ b_g, [\![x]\!]$ \\
\textbf{Output:} $[\![x_{next}]\!]$
\end{flushleft}

\hrulefill
\begin{algorithmic}[1]
\Function {\textsc{Traverse}}{$ b_g, [\![x]\!]$}
	\If {$b_g = 0$} \label{Traverse_To_Left} \Comment{traverse to left}
		\State Let $[\![x_{next}]\!]$ be the left child node of $[\![x]\!]$ 
	\Else \label{Traverse_To_Right} \Comment{traverse to right}
		\State Let $[\![x_{next}]\!]$ be the right child node of $[\![x]\!]$ 
	\EndIf
	\State \textbf{return} $[\![x_{next}]\!]$
\EndFunction
\end{algorithmic}
\label{TreeTraversal_Algorithm}
\end{algorithm}

\textbf{Encryption Algorithm.} The encryption algorithm (Algorithm \ref{Encryption_Algorithm}) runs at the CSP as well and is called only if the tree traversal (Algorithm \ref{TreeTraversal_Algorithm}) has to stop. Then the compared values are strictly ordered and depending on that the algorithm finds the closest element to the current node in the OPE-table. This element is either the predecessor if DA's input is smaller (line \ref{Encrypt_With_Pred}) or the successor if DA's input is larger (line \ref{Encrypt_With_Succ}). Then if necessary (line \ref{Balance_Condition}) rebalance the tree and compute the ciphertext as in line \ref{Return_New_Ciphertext}.
\begin{algorithm}
\caption{Encryption for value $\overline{x}$}

\begin{algorithmic}[1]

\Function {\textsc{Encrypt}}{$b_g, [\![x]\!]$}
\State \textbf{retrieve} $\left\langle [\![x]\!], y \right\rangle$ in the OPE-table
	\If {$b_g = 0$} \label{Encrypt_With_Pred} 
		\State \textbf{retrieve} predecessor $\left\langle [\![x']\!], y' \right\rangle$ of $\left\langle [\![x]\!], y \right\rangle$ 
		\State \textbf{let} $y_l \gets y'$ \textbf{and} $y_r \gets y$ \Comment{$y'<y$}
	\Else \label{Encrypt_With_Succ} 
		\State \textbf{retrieve} successor $\left\langle [\![x'']\!], y'' \right\rangle$ of $\left\langle [\![x]\!], y \right\rangle$ 
		\State \textbf{let} $y_l \gets y$ \textbf{and} $y_r \gets y''$ \Comment{$y<y''$}
	\EndIf
	
	\If {$y_r - y_l = 1$} \label{Balance_Condition}
		\State \textbf{rebalance} the OPE-tree
	\EndIf
	\State $\overline{y} \gets y_l + \left\lceil \frac{y_{r} - y_{l}}{2} \right\rceil$ \label{Return_New_Ciphertext}
\State \textbf{return} $\overline{y}$
\EndFunction
\end{algorithmic}
\label{Encryption_Algorithm}
\end{algorithm}

\subsection{Correctness and Security Proofs}
The security of Yao's protocol is proven in \cite{LindellP.2009} and provides a simulator, that will be used to construct simulators for DO and DA.


\begin{theorem}[Correctness]
The protocol $\Pi$\textsubscript{\tiny \normalfont OOPE} is correct.
\end{theorem}
\begin{proof}
Let $b_g=(if ~ \overline{x}>x ~ then ~ 1 ~ else ~ 0)$ and $b_e = (if ~ \overline{x} \neq x ~ then ~ 1 ~ else ~ 0)$. 
From inputs $(b_a, b'_a, \overline{x}+r)$ of the DA and $(b_o, b'_o,\overline{x}+r)$ of the DO the garbled circuit returns $(b_e\oplus b_a\oplus b_o, b_g\oplus b'_a \oplus b'_o)$ to the DA and $(b_e\oplus b_a\oplus b_o, b_g\oplus b'_a\oplus b'_o)$ to the DO. Then the DA resp. the DO sends $(b_a, b'_a, b_e\oplus b_o, b_g \oplus b'_o)$ resp. $(b_o, b'_o, b_e\oplus b_a, b_g\oplus b'_a)$ to the CSP. With Equation \ref{Eqn_GC_Result_Check} the CSP can correctly deduce $b_e$ and $b_g$. 
\end{proof}

\begin{theorem}[Security] 
\label{Theorem_OOPE_Security}
The protocol $\Pi$\textsubscript{\tiny \normalfont OOPE} securely implements the OOPE functionality in the semi-honest model with honest majority.
\end{theorem}

\subsection{Dealing with a malicious activities}
As above we assume honest majority. In fact malicious DA or DO can only cheat in Yao's protocol or by returning a fake output of the comparison step to the CSP. Results of the comparison step can be checked with Equation \ref{Eqn_GC_Result_Check}. For cheating in Yao's protocol, there are solutions based on the cut-and-choose technique that deal with malicious parties \cite{LindellP.2007, LindellR.2014, Lindell.2016}. So in this section we concentrate on the malicious CSP. Recall that the CSP holds the OPE-table that is a set of ciphertexts $\left\langle [\![x]\!], y \right\rangle$, where $[\![x]\!]$ is a Paillier ciphertext under public key $pk$. Our goal is to prevent a malicious CSP to replace the $\left\langle [\![x]\!], y \right\rangle$ with self-generated ciphertexts $\left\langle [\![\tilde{x}]\!], \tilde{y} \right\rangle$. The solution consists in computing for each ciphertexts $\left\langle [\![x]\!], y \right\rangle$ a message authentication code (MAC) that will constraint the CSP to use valid $[\![x]\!]$ in the OOPE protocol.

\noindent\textbf{Discrete Logarithm.} Our first solution is to use discrete logarithm. In the initialization step the DO computes for each node $[\![x]\!]$ a MAC $(g^x \bmod p)$ where $g$ and the prime number $p$ are unknown to the CSP. Then in each comparison step the CSP sends $[\![x+r]\!]$ to the DO and $(r, g^x \bmod p)$ to the DA. The DO reveals $g$, $p$ and $m = (g^{x+r} \bmod p)$ to the DA after decryption. Finally the DA computes $(g^x g^r \bmod p) = (g^{x+r} \bmod p)$ and checks if it is the same as $m$. Only if the check succeeds they engage in the oblivious comparison protocol. Besides the OOPE protocol the DO can also use the discrete logarithm solution to check the integrity of its database. 

\noindent\textbf{Pedersen Commitment.} The above solution with discrete logarithm is not perfectly hiding and only based on the fact that solving discrete logarithm is computationally difficult.~However if the space of possible value of $x$ is small, anyone who knows $g$ and $p$ could simply try them all. Hence, this solution hide $x$ only to the CSP, because the CSP does not know $g$ and $p$. However, the solution is vulnerable to the DA. To perfectly hide $x$ the hiding scheme must be semantically secure. We therefore propose a solution based on Pedersen commitment \cite{Pedersen.1991}. The DO chooses $g$ and $p$ as above and another number $h$ and reveals them only to the DA. In the initialization step the DO stores each node $[\![x]\!]$ with the pair $([\![a]\!], g^xh^a \bmod p$), where $a$ is random. In the first step of the comparison protocol the CSP chooses two random numbers $r$, $r'$, computes and sends $c_x = [\![x+r]\!]$, $c_a = [\![a+r']\!]$ to the DO. Then the CSP sends $(g^xh^a \bmod p)$, $r$ and $r'$ to the DA. Next the DO decrypts both ciphertexts $c_x, c_a$ and sends $m = (g^{x+r}h^{a+r'} \bmod p)$ to the DA. Finally the DA computes $m' = (g^xh^a * g^rh^{r'} \bmod p) = (g^{x+r}h^{a+r'} \bmod p)$ and checks if $m$ and $m'$ are equal.

Notice that, a malicious CSP can still corrupt the homomorphic ciphertext in the OPE-table, as homomorphic encryption is malleable. Additionally, the server may also alter its responses to the client in an attempt to learn additional information on top of the order of encrypted values \cite{PopaLZ.2013}. To force the server to perform these operations correctly, we can adapt the idea of \cite{PopaLZ.2013} that consists of adding Merkle hashing on top of the OPE-tree and to use it to check the correctness of the server’s responses.

\begin{figure}
\centering
\resizebox{0.5\textwidth}{!} {
\begin{tikzpicture}[>=triangle 45,
    square/.style={draw, minimum size=15mm,
    }]
		\draw
			node[square, label={[anchor=east]25:$c_{e,0}$}, label={[anchor=east]-25:$c_{g,0}$} ] (A0) {} 
			node[square, left=5mm of A0, label={[anchor=east]25:$c_{e,1}$}, label={[anchor=east]-25:$c_{g,1}$}] (A1) {}
			node[transparent=yes, square, left=5mm of A1] (AF) {}
			node[draw=none] (space) [left= 0.5cm of A1]  {$~~~$}
			node[draw=none] (spacea) [above= 0.8mm of space]  {$...$}
			node[draw=none] (spaceb) [below= 0.8mm of space]  {$...$}
			node[square, left=5mm of space, label={[anchor=east]25:$c_{e,l-1}$}, label={[anchor=east]-25:$c_{g,l-1}$}] (Al) {}
			node[transparent=yes, square, right=5mm of Al] (AG) {}
			node[transparent=yes, square, right=5mm of A0] (A00) {}
			node[square, left=5mm of Al, label={[anchor=east]25:$c_{e,l}$}, label={[anchor=east]-25:$c_{g,l}$}] (AO) {$\oplus$}
			node[transparent=yes, square, left=5mm of AO, label={[anchor=east]25:$c_{e}$}, label={[anchor=east]-25:$c_{g}$}] (AH) {}
			node[draw=none] (spaceAO1) [below= 0.5cm of AO.245]  {$b'_{x}$}
			node[draw=none] (spaceAO2) [below= 0.5cm of AO.295]  {$b'_{\overline{x}}$}
			;

    \draw[<-] (A1.115) --++(90:0.5cm) node [above] {$x_1$};
    \draw[<-] (A1.65) --++(90:0.5cm) node [above] {$\overline{x}_1$};
    \draw[<-] (A0.115) --++(90:0.5cm) node [above] {$x_0$};
    \draw[<-] (A0.65) --++(90:0.5cm) node [above] {$\overline{x}_0$};
		
		\draw[<-] (Al.115) --++(90:0.5cm) node [above] {$x_{l-1}$};
    \draw[<-] (Al.65) --++(90:0.5cm) node [above] {$\overline{x}_{l-1}$};
		
		\draw[<-] (AO.115) --++(90:0.5cm) node [above] {$b_{x}$};
    \draw[<-] (AO.65) --++(90:0.5cm) node [above] {$b_{\overline{x}}$};
		
		\draw[->] (spaceAO1) -- (AO.245) node [below] {};
		\draw[->] (spaceAO2) -- (AO.295) node [below] {};

    \draw[->] (A00.155) -- (A0.25) node [right = 0.45cm] {0};
		\draw[->] (A00.205) -- (A0.-25) node [right = 0.45cm] {0};
		
		\draw[->] (A0.155) -- (A1.25);
		\draw[->] (A0.205) -- (A1.-25);
		
		\draw[->] (AG.155) -- (Al.25);
		\draw[->] (AG.205) -- (Al.-25);

		\draw[->] (A1.155) -- (AF.25) ;
		\draw[->] (A1.205) -- (AF.-25) ;
		
		\draw[->] (Al.155) -- (AO.25) ;
		\draw[->] (Al.205) -- (AO.-25) ;
		
		\draw[->] (AO.155) -- (AH.25) ;
		\draw[->] (AO.205) -- (AH.-25) ;
		
\end{tikzpicture}
}
\caption{ Overview of the Garbled Circuit GC$_{=, >}$ for comparison and equality test. \textmd{Each box for $ i = 0, \ldots, l-1$ is a 1-bit circuit for equality and greater than test and outputs $ c_{e, j+1} = (x_j \oplus \overline{x}_j) \vee c_{e, j}$ (Equation \ref{Eqn_Equality_Test}) and $c_{g,j+1} =(x_j \oplus c_{g, j}) \wedge (\overline{x}_j \oplus c_{g, j}) \oplus x_j$ (Equation \ref{Eqn_Greater_Than_Test}) resp.. The last circuit implements exclusive-or operation and outputs $c_e = c_{e, l} \oplus b_x \oplus b_{\overline{x}} $ and $c_g =c_{g, l} \oplus b'_x \oplus b'_{\overline{x}}$.}} 
\label{GarbledCircuit}
\end{figure}
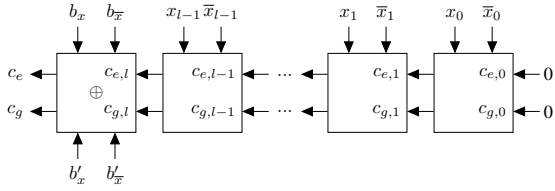
\section{Protocol for Integer Comparison}
\label{Sec_Int_Comparison}
For our oblivious OPE protocol we needed garbled circuit for comparison and equality test and adapted the garbled circuits of \cite{KolesnikovS.2008, KolesnikovSS.2009} to our needs. Firstly, instead of implementing one garbled circuit for comparison and another one for equality test, we combined both in the same circuit. This allows to use the advantage that almost the entire cost of garbled circuit protocols can be shifted into the setup phase. In Yao's protocol the setup phase contains all expensive operations (i.e., computationally expensive OT and creation of GC, as well as the transfer of GC that dominates the communication complexity) \cite{KolesnikovSS.2009}. Hence, by implementing both circuits in only one we reduce the two costly setup phases to one as well. Secondly, in our oblivious OPE protocol, integer comparison is an intermediate step, hence the output should not be revealed to the parties participating in the protocol, since this will leak information. Thus the input of the circuit contains a masking bit for each party that is used to mask the actual output. Only the party that receives the masked output and both masking bits can therefore recover the actual output. Let GC$_{=,>}$ denote this circuit.

Let $P_1$, $P_2$ be party one and two resp. and let $x = x_{l-1}, \ldots, x_{0}$, $\overline{x} = \overline{x}_{l-1}, \ldots, \overline{x}_{0}$ be their respective inputs in binary representation. Parties $P_1$ and $P_2$ choose masking bits $b_x$, $b'_x$, $b_{\overline{x}}$, $b'_{\overline{x}}$ and extend their input to $(b_x, b'_x, x_{l-1}, \ldots, x_{0})$, $(b_{\overline{x}}, b'_{\overline{x}}, \overline{x}_{l-1}, \ldots, \overline{x}_{0})$ respectively. An overview of the circuit is illustrated in Figure \ref{GarbledCircuit}.

For equality test  we use Equation\footnote{In $c_{e, j}$ and $c_e$, $e$ stands for \textbf{equality test} and $j$ is the bit index} \ref{Eqn_Equality_Test}. The two first lines are from \cite{KolesnikovS.2008} and test from $0$ to $l-1$ if the bits are pairwise different (i.e their exclusive-or is 1). If not we use the result of the previous bit test. Initially, this bit is set to $0$.
\begin{equation}
  \begin{cases}
    c_{e, 0}   &= 0  \\
    c_{e, j+1} &= (\overline{x}_j \oplus x_j) \vee c_{e, j}, ~ j = 0, \ldots, l-1 \\
    c_e        &= c_{e, l} \oplus b_x \oplus b_{\overline{x}} 
  \end{cases}
	\label{Eqn_Equality_Test}
\end{equation}
The actual output of the circuit $c_{e, l}$ is $1$ if $x$ and $\overline{x}$ are different and $0$ otherwise (i.e. $c_{e, l} = \left[\overline{x}\neq x\right]?1:0$). Then we blind $c_{e, l}$ by applying exclusive-or operations with the masking bits $b_x$ and $b_{\overline{x}}$.

The comparison functionality is defined as (if $\overline{x}>x$ then $1$ else $0$) (i.e $\left[\overline{x}>x\right]?1:0$). In \cite{KolesnikovSS.2009} the circuit is based on the fact that $\left[\overline{x}>x\right] \Leftrightarrow \left[\overline{x} - x - 1 \geq 0\right]$ and is summarized in Equation\footnote{In $c_{g, j}$ and $c_g$, $g$ stands for \textbf{greater than} and $j$ is as above} \ref{Eqn_Greater_Than_Test}, where again the two first are from \cite{KolesnikovSS.2009}. The second line represents the 1-bit comparator which depends on the previous bit comparison. This is initially $0$.
\begin{equation}
  \begin{cases}
    c_{g,0}   &=0 \\
    c_{g,j+1} &=(\overline{x}_j \oplus c_{g, j}) \wedge (x_j \oplus c_{g, j}) \oplus \overline{x}_j, ~ j = 0, \ldots, l-1 \\
    c_g       &=c_{g, l} \oplus b'_x \oplus b'_{\overline{x}} 
  \end{cases}
	\label{Eqn_Greater_Than_Test}
\end{equation}
Again the actual output $c_{g, l}$ is blinded by applying exclusive-or operations with the masking bits $b'_x$ and $b'_{\overline{x}}$.

\section{OOPE with Frequency-hiding OPE}
\label{Sec_OOPE_Protocol_FHOPE}
In this section we consider the case where the underlying OPE is not deterministic as in \cite{Kerschbaum.2015}. As above the first step is the initialization procedure (Section \ref{Initialization}). It remains the same with the difference that the tree traversal and the encryption algorithms work as in Kerschbaum's scheme mOPE$_3$. Hence if the equality test returns true (line \ref{Traverse_Stops_With_Equality} of Protocol \ref{OOPE_Protocol}), the CSP chooses a random coin and then traverses the tree to the left or to the right depending on the outcome of the coin. The order $\overline{y}$ of $\overline{x}$ is computed as $\overline{y} = y_{i-1} + \lceil \frac{y_{i} - y_{i-1}}{2} \rceil$ resp. $\overline{y} = y_{i} + \lceil \frac{y_{i+1} - y_{i}}{2} \rceil$ if the algorithm is inserting $\overline{x}$ left resp. right to a node $[\![x_{i}]\!]$ with corresponding order $y_{i}$. However, the equality test leaks some information, as it allows the CSP to deduce from the OPE-table that certain nodes have the same plaintext. Therefore it would be preferable to implement the random coin in the secure computation.
\subsection{Implementing the random coin in garbled circuit}
In the following $\overline{x}$ and $x$ represent as before the inputs of the DA and the DO in the oblivious comparison respectively, and GC$_{=, >}^{u}$ represents the unmasked comparison circuit\footnote{This is the sub-circuit of Figure \ref{GarbledCircuit} that operates on the real input bits (from 0 to $l-1$) without the masking bits.} that outputs the bits $b_e$ as result of the equality test and $b_g$  as result of the greater than comparison. The idea is to adapt the garbled circuit for integer comparison (Section \ref{Sec_Int_Comparison}) such that its output allows to traverse the tree randomly as in \cite{Kerschbaum.2015}, but without revealing the result of the equality test to the CSP.
\begin{lemma}
Let $r_{\overline{x}}$ and $r_{x}$ be some DA's and DO's random bits and $b_{r} = r_{\overline{x}} \oplus r_{x}$. Then extending the circuit GC$_{=, >}^{u}$ to the circuit GC$_{b}^{u}$ with additional input bits $r_{\overline{x}}$, $r_{x}$ and with output $b = (b_{e} \wedge b_{g}) \vee (\neg b_{e} \wedge b_{r}) $ traverses the tree as required.
\end{lemma}
\begin{proof}
If $\overline{x} \neq x$ then $b_e = 1$ and $b = b_g$, hence the algorithm traverses the tree depending on the greater than comparison. Otherwise $\neg b_{e} = 1$, hence $b$ is the random bit $b_r$ and the tree traversal depends on a random coin. In each case the circuit returns either 0 or 1, and does not reveal if the inputs are equal.
\end{proof}
Now the circuit GC$_{b}^{u}$ can also be extended to the circuit GC$_{b}$ by using the masking bits $b_a$ and $b_o$ for the DA and the DO respectively as described in Section \ref{Sec_Int_Comparison}. The output is then $((b_{g} \wedge b_{e}) \vee (\neg b_{e} \wedge b_{r}))\oplus b_a\oplus b_o$.

However, care has to be taken when returning the random bit. Recall that the protocol loops $h$ times to prevent 
the DA and DO from learning the right number of comparisons. Hence if we reach equality before having performed $h$
comparisons the garbled circuit computation must keep returning the same random bit to prevent leaking that information
to the CSP. Therefore DA and DO must keep track on shares of $b_e$ and $b_r$ which are extra inputs to the circuit. Let $\hat{b}_e, \hat{b}_r$ be the previous equality bit (initially 1, e.g. 0 for DO and 1 for DA) and random bit (initially 0), then the garbled circuit must execute the following procedure: If $b_e = 0$ then check if $\hat{b}_e = 0$ and return $\hat{b}_r$ otherwise return $b_r$. If $b_e \neq 0$ then return $b_g$.

\subsection{Dealing with queries}
So far we have computed the ciphertext in the non-deterministic case. However, as Kerschbaum pointed out \cite{Kerschbaum.2015} this ciphertext cannot be directly used to query the database. As in the deterministic case let $x$ and $y$ be symbols for plaintext and order respectively. Since a plaintext $x$ might have many ciphertexts let $c_{min}$ and $c_{max}$ be respectively the minimum and maximum order of $x$, hence:
\begin{equation}
  \begin{cases}
		c_{min}(x) = \mbox{min}(\{y:\textsc{Decrypt}\textsubscript{\tiny OPE}(\langle [\![x]\!], y \rangle, sk)=x\})  \\
		c_{max}(x) = \mbox{max}(\{y:\textsc{Decrypt}\textsubscript{\tiny OPE}(\langle [\![x]\!], y \rangle, sk)=x\})
  \end{cases}
	\label{Eqn_FHOPE_Min_Max_Cipher}
\end{equation}

Thus, a query $[a, b]$ must be rewritten in $[c_{min}(a), c_{max}(b)]$. Unfortunately, in Kerschbaum's scheme the $c_{min}(x), c_{max}(x)$ are only known to the DO, because they reveal to the server the frequency of plaintexts. Recall that the goal of \cite{Kerschbaum.2015} was precisely to hide this frequency from the CSP.

Instead of returning $\overline{y}$ to the DA, which is useless for queries, our goal is to allow the DA to learn $c_{min}(\overline{x})$ and $c_{max}(\overline{x})$ and nothing else. The CSP learns only $\langle [\![\overline{x}]\!], \overline{y} \rangle$ as before and the DO learns nothing besides the intermediate messages of the protocol. We begin by proving the following lemma.
\begin{lemma}
Let $y_{i}$, $y_{i+1}$ be the order (i.e. $y_{i}<y_{i+1}$) of already encrypted plaintexts $x_{i}, x_{i+1}$ (i.e. $x_{i} \leq x_{i+1}$). Let $x$ be a new plaintext with corresponding order $y$ such that $x_{i} \leq x \leq x_{i+1}$. Then it holds: $c_{min}(x) \in \{c_{min}(x_{i}), y\}$ and $c_{max}(x) \in \{c_{max}(x_{i+1}), y\}$.
\label{New_Min_Max_Order_Lemma}
\end{lemma}
\begin{proof}
If $x_{i} = x$ then by definition of $c_{min}$ we have $c_{min}(x) = c_{min}(x_{i})$. If $x_{i} < x$ and $x < x_{i+1}$ then $x$ occurs only once in the data set and it holds $c_{min}(x) = c_{max}(x) =y$. Otherwise $x$ is equal to $x_{i+1}$, but since $x$ is new and by assumption $x \leq x_{i+1}$ the algorithm is inserting $x$ right to $x_{i}$ and left to $x_{i+1}$ hence $y_{i} < y < y_{i+1}$ must hold. Then by definition again $c_{min}(x) = y$. \\ For the case of max the proof is similar.
\end{proof}

\begin{corollary}
Let $x, x_{i}, x_{i+1}, y, y_{i}, y_{i+1}$ be as above and let $b_{i} = [x_{i}=x]?1:0$ resp. $b_{i+1}=[x=x_{i+1}]?1:0$ then it holds: 
$c_{min}(x) = b_{i} \cdot c_{min}(x_{i}) + (1-b_{i})\cdot y$, resp. $c_{max}(x) = b_{i+1} \cdot c_{max}(x_{i+1}) + (1-b_{i+1})\cdot y$.
\label{New_Min_Max_Order_Formula}
\end{corollary}

Now we are ready to describe the solution. First we assume that tree rebalancing never happens, because it might update $c_{min}$ and $c_{max}$ for some ciphertexts. The CSP cannot update $c_{min}$ and $c_{max}$ without knowing the frequency. According to \cite{Kerschbaum.2015} the probability of rebalancing is negligible in $n$ for uniform inputs if the maximum order $M$ is larger than $2^{6.4\cdot \log_2n}$. For non-uniform input, smaller values of $M$ are likely.

The first step is to store besides each ciphertext $\langle [\![x]\!], y \rangle$ two ciphertexts $[\![c_{min}(x)]\!]$ and $[\![c_{max}(x)]\!]$. This is done by the DO during the initialization. Let $[\![x]\!]_{DA}$ be a ciphertext of $x$ encrypted with a Paillier public key, whose corresponding private key belongs to the DA. After the computation of $\overline{y}$ (Protocol \ref{OOPE_Protocol}) the CSP learns $\langle [\![\overline{x}]\!], \overline{y} \rangle$. Then the parties execute Protocol \ref{Min_Max_Order_Protocol} with $\langle\mathbb{S}, \langle [\![\overline{x}]\!], \overline{y} \rangle\rangle$ and $sk$ as input for the CSP and the DO respectively. The DA does not have any input, but receives alone the output of the protocol.
\begin{algorithm}[t]
\floatname{algorithm}{Protocol}
\caption{Min Max Order Protocol}
\begin{flushleft}
\textbf{Input ($In_{CSP}, In_{DA}, In_{DO}$)}: $(\langle\mathbb{S}, \langle [\![\overline{x}]\!], \overline{y} \rangle\rangle, \emptyset, sk)$ \\
\textbf{Output ($O_{CSP}, O_{DA}, O_{DO}$)}: $(\emptyset, \left\langle c_{min}(\overline{x}), c_{max}(\overline{x}) \right\rangle, \emptyset)$ \\
\textbf{Functionality }: \textsc{MinMaxOrder}($\mathbb{S}, \langle [\![\overline{x}]\!], \overline{y} \rangle, sk$)
\end{flushleft}
\hrulefill
\begin{algorithmic}[1]
	\State CSP: retrieve  $\langle [\![x_{i}]\!], y_{i} \rangle, \langle [\![x_{i+1}]\!], y_{i+1} \rangle$ s.t. $y_{i} < \overline{y} < y_{i+1}$
	\State CSP: choose random integers $s_{1}$, $s_{2}$, $r_1$, $r_2$
	\State CSP: compute $[\![d_{1}]\!] \leftarrow [\![(x_{i}-\overline{x})\cdot s_{1}]\!]$
	\State CSP: compute $[\![d_{2}]\!] \leftarrow [\![(x_{i+1}-\overline{x})\cdot s_{2}]\!]$
	\State CSP $\rightarrow$ DO: send $\langle [\![d_{1}]\!], [\![\overline{y}\cdot r_{1}]\!]_{DA}, [\![c_{min}(x_{i})\cdot r_{1}]\!] \rangle$
	\State CSP $\rightarrow$ DO: send $\langle [\![d_{2}]\!], [\![\overline{y}\cdot r_{2}]\!]_{DA}, [\![c_{max}(x_{i+1})\cdot r_{2}]\!] \rangle$
	\State CSP $\rightarrow$ DA: send $r_1$ and $r_2$
	\State DO: decrypt $[\![d_{1}]\!]$, $[\![d_{2}]\!]$ and evaluate Equation \ref{Eqn_FHOPE_Min_Max_Oblivious_Cipher} 
	\State DO $\rightarrow$ DA: send $[\![c_{min}(\overline{x})\cdot r_{1}]\!]_{DA}$ 
	\State DO $\rightarrow$ DA: send $[\![c_{max}(\overline{x})\cdot r_{2}]\!]_{DA}$
	\State DA: decrypt and output $c_{min}(\overline{x})$ and $c_{max}(\overline{x})$
\end{algorithmic}
\label{Min_Max_Order_Protocol}
\end{algorithm}
\begin{equation}
  \begin{cases}
	b_1 \leftarrow [d_{1}=0], b_2 \leftarrow [d_{2}=0] \\
	[\![c_{min}(\overline{x})\cdot r_{1}]\!]_{DA} \leftarrow b_{1}?[\![c_{min}(x_{i})\cdot r_{1}]\!]_{DA}:[\![\overline{y}\cdot r_{1}]\!]_{DA} \\
	[\![c_{max}(\overline{x})\cdot r_{2}]\!]_{DA} \leftarrow b_{2}?[\![c_{max}(x_{i+1})\cdot r_{2}]\!]_{DA}:[\![\overline{y}\cdot r_{2}]\!]_{DA}
  \end{cases}
	\label{Eqn_FHOPE_Min_Max_Oblivious_Cipher}
\end{equation}
Notice that for an input $\overline{x}$ of the DA the ciphertext $\langle [\![\overline{x}]\!], \overline{y} \rangle$ is not inserted in the database, but only in the OPE-table, because it cannot be included in the result of a query. Particularly, if $\langle [\![\overline{x}]\!], \overline{y} \rangle$ is no longer needed (e.g.: after the data analysis) it must be removed from the OPE-table. As stated in Lemma \ref{New_Min_Max_Order_Lemma}, if it happens that the new $\overline{x}$ with corresponding order $\overline{y}$ is inserted between $x_{i}$ and $x_{i+1}$ such that $x_{i} < \overline{x} = x_{i+1}$ then $c_{min}(\overline{x}) = \overline{y}$ implies that the previous $c_{min}(x_{i+1})$ should be updated to $\overline{y}$. However, as explained before this update is not necessary.

In Protocol \ref{Min_Max_Order_Protocol} the DO sees two semantically secure ciphertexts $[\![\overline{y}\cdot r_{1}]\!]_{DA}$ and $[\![\overline{y}\cdot r_{2}]\!]_{DA}$, which it cannot decrypt, and four randomized plaintexts $d_{1}, c_{min}(x_{i})\cdot r_{1}, d_{2}, c_{max}(x_{i+1})\cdot r_{2}$. The DA sees two random integers $r_1, r_2$ and the output of the protocol. The CSP receives no new message. Hence simulating the protocol is straightforward.

\section{Implementation}
\label{Sec_Implementation}
We have implemented our scheme using SCAPI (Secure Computation API)\cite{EjgenbergFLL.2012}. SCAPI is an open-source Java library for implementing secure two-party and multiparty computation protocols. It provides a reliable, efficient, and highly flexible cryptographic infrastructure. It also provides many optimizations of garbled circuits construction such as OT extensions, free-XOR, garbled row reduction \cite{EjgenbergFLL.2012}. Furthermore, there is a built-in communication layer that provides communication services for any interactive cryptographic protocol. This layer is comprised of two basic communication types: a two-party communication channel and a multiparty layer that arranges communication between multiple parties.

\subsection{Parameters} 
The first parameter that should be defined for the experiment is the security parameter (i.e. bit length of the public key) of Paillier's scheme (e.g. 2048 or 4096). Paillier's scheme requires to choose two large prime numbers $P$ and $Q$ of equal length and to compute a modulus $N = PQ$ and the private key $\lambda = lcm(P-1, Q-1)$. Then select a random $g \in \mathbb{Z}_{N^2}^{*}$ such that if $e$ is the smallest integer with $g^e = 1\bmod N^2$, then $N$ divides $e$. The public key is $(g, N)$. To encrypt a plaintext $m$ select a random $r \in \mathbb{Z}_{N}^{*}$ and compute Equation \ref{Paillier_Encrypt}. To decrypt a ciphertext $c$ compute Equation \ref{Paillier_Decrypt} with $L(u) = \frac{u-1}{N}$ and $\mu = (L(g^{\lambda}\bmod N^2))^{-1}\bmod N$.
\begin{eqnarray}
c &\leftarrow& g^m r^N\bmod N^2 \label{Paillier_Encrypt}\\
m &\leftarrow& L(c^{\lambda}\bmod N^2) \cdot \mu\bmod N \label{Paillier_Decrypt}
\end{eqnarray}

The other parameters of the OOPE protocol are the length of the inputs (e.g. 32, 64, 128, 256 bits integer), the length of the order $\log_2M$ - with $M$ the maximal order - (e.g. 32, 64, 128 bits), and the size of the database (e.g. $10^3$, $10^4$, $10^5$, $10^6$ entries).

\subsection{Optimization} 
To reduce the execution cost of our scheme we applied optimizations of Paillier's scheme as recommended in \cite{Paillier.1999}. 
We implemented our scheme with $g = 1+N$. This transforms the modular exponentiation $g^m\bmod N^2$ to a multiplication, since $(1+N)^m \bmod N^2 = 1+mN \bmod N^2$. 
Moreover, we precomputed $\mu$ in Equation \ref{Paillier_Decrypt}, used Chinese remaindering for decryption and pre-generated randomness for encryption and homomorphic plaintext randomization (Protocol \ref{Comparison_Protocol}). As a result, encryption, decryption and homomorphic addition take respectively 52$\mu s$, 12$ms$ and 67$\mu s$ when the key length is 2048 bits.

\subsection{Evaluation}
To evaluate the performance of our scheme we answer the following questions:
\begin{itemize}
	\item What time does the scheme take to encrypt an input of the DA?
	\item How does the network communication influence the protocol?
	\item What is the average generation time and the storage cost of the OPE-tree?
\end{itemize}

\textbf{Experimental Setup.} We chose 2048 bits as security parameter for Paillier's scheme and ran experiments via loopback address and via LAN using 3 machines with Intel(R) Xeon(R) CPU E7-4880 v2 at 2.50GHz. For the LAN experiment, the first machine with 4 CPUs and 8 GB RAM ran the CSP, the second machine with 4 CPUs and 4 GB RAM ran the DO and the last machine with 2 CPUs and 2 GB RAM ran the DA. For the loopback experiment we used the first machine. \\
We generated the OPE-tree with random inputs, balanced it and encrypted the plaintexts with Paillier encryption. For the DA, we generated 100 random inputs. Then we executed the OOPE protocol 100 times and computed the average time spent in the overall protocol, in the oblivious comparison, in Yao's protocol, in Paillier's decryption.

\begin{figure}[t]
\centering
\includegraphics[width=0.40\textwidth]{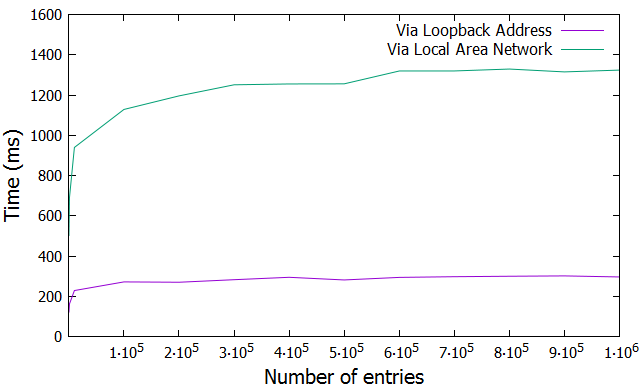}
\caption{Encryption Cost of OOPE} 
\label{Graph_Encryption_Cost}
\end{figure}

\begin{figure}
    \centering
    \begin{subfigure}[b]{0.40\textwidth}
        \includegraphics[width=\textwidth]{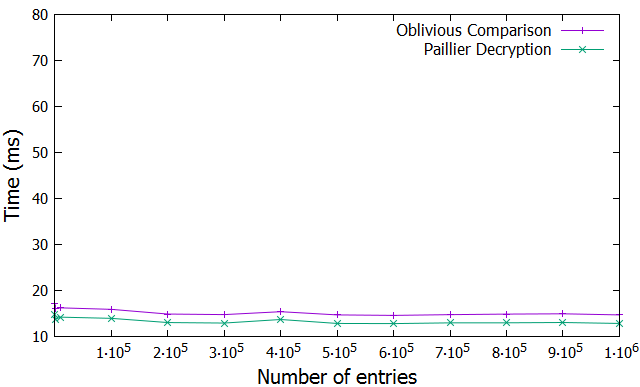}
        \caption{Via Loopback}
        \label{Graph_CMP_Cost_LOOP}
    \end{subfigure}
    ~
    \begin{subfigure}[b]{0.40\textwidth}
        \includegraphics[width=\textwidth]{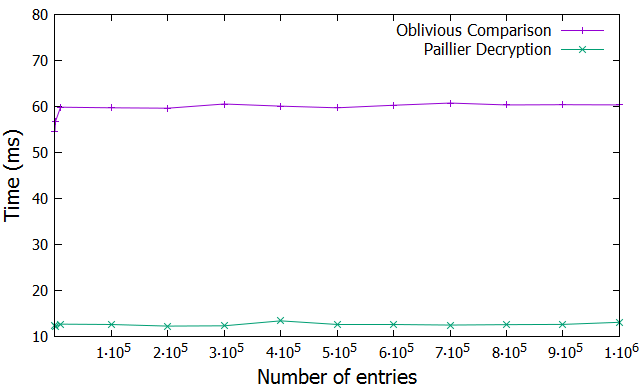}
        \caption{Via LAN}
        \label{Graph_CMP_Cost_LAN}
    \end{subfigure}
    \caption{Cost of oblivious comparison}\label{Graph_CMP_Cost}
\end{figure}

\textbf{Encryption Cost.} Figure \ref{Graph_Encryption_Cost} shows the average cost (y-axis) needed to encrypt a value with the OOPE protocol for database sizes (x-axis) between 100 and 1,000,000. Overall, the cost for OOPE goes up as the size of the database increases. This is because the depth of the tree increases with its size. Hence, this implies larger number of oblivious comparisons for larger trees. The average encryption time of OOPE for a database with one million entries is about 0.3 s via loopback (1.3 s via LAN). This cost corresponds to the cost of comparison multiply by the number of comparisons (e.g. 20 comparisons for 1000000 entries). \\
The inherent sub-protocol for oblivious comparison does not depend on the database size but on the input length and the security parameter $\log_2 N$. Figure \ref{Graph_CMP_Cost} shows that this cost is almost constant for each database size. Via loopback (Figure \ref{Graph_CMP_Cost_LOOP}) the comparison costs about 14 ms which is dominated by the time (about 12 ms to the DO) to decrypt $[\![x+r]\!]$ in step \ref{Compare_CSP_To_DO} of Protocol \ref{Comparison_Protocol}. The remaining 2 ms are due to the garbled circuit execution, since the overhead due to network communication is negligible. Figure \ref{Graph_CMP_Cost_LAN} shows how the network communication affects the protocol. Via LAN (Figure \ref{Graph_CMP_Cost_LAN}) the comparison costs about 60 ms where the computation is still dominated
by the 12 ms for decryption. However, the network traffic causes an overhead of about 46 ms.

\begin{figure}
    \centering
    \begin{subfigure}[b]{0.40\textwidth}
        \includegraphics[width=\textwidth]{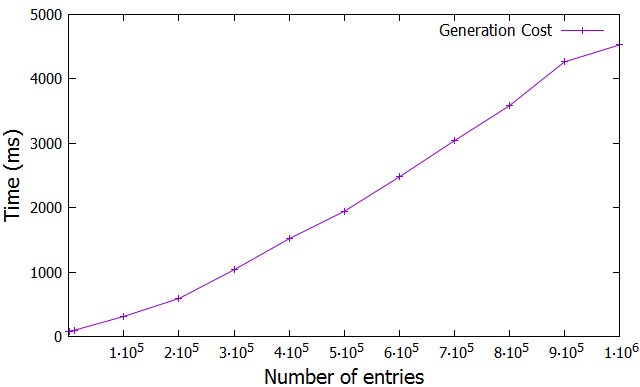}
        \caption{Average generation time}
        \label{OPE_Tree_Cost_Time}
    \end{subfigure}
		~
    \begin{subfigure}[b]{0.40\textwidth}
        \includegraphics[width=\textwidth]{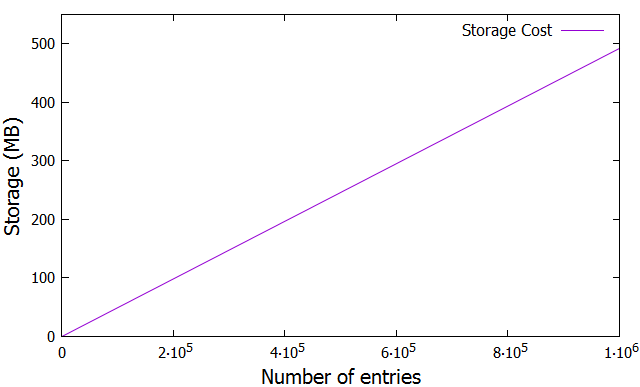}
        \caption{Storage cost for $\log_2N = 2048$}
        \label{OPE_Tree_Cost_Mem}
    \end{subfigure}
    \caption{OPE-tree costs}\label{OPE_Tree_Cost}
\end{figure}

\textbf{OPE-tree costs.} The time to generate the OPE-tree also increases with the number of entries in the database and it is dominated by the time needed to encrypt the input data with Paillier's scheme. However, the above optimizations (i.e. choice of $g = 1+N$ and pre-generated randomness) enable a fast generation of the OPE-tree. Figure \ref{OPE_Tree_Cost_Time} illustrates the generation time on the y-axis for databases with size between 100 and 1,000,000 on the x-axis. For 1 million entries, the generation costs on average only about 4.5 seconds. \\
The storage cost of the tree depends on $\log_2 N$, the bit length of the order and the database size. Since Paillier ciphertexts are twice longer than $\log_2 N$,  each OPE ciphertext $\langle [\![x]\!], y\rangle$ needs $2\cdot\log_2N+\log_2M$ bits storage. This is illustrated in Figure \ref{OPE_Tree_Cost_Mem}, with the x-axis representing the database size. The scheme needs 492.1 MB to store 1 million OPE ciphertexts, when the security parameter is 2048 and the order is 32-bit long.

\section{Conclusion}
\label{Sec_Conclusion}
Since order-preserving encryption (OPE) schemes are limited to the use case to one server and one client, we introduced a novel notion of oblivious order-preserving encryption (OOPE) as an equivalent of a public-key order-preserving encryption. Then we presented a protocol for OOPE that combines deterministic OPE schemes based on binary tree search with Paillier's homomorphic encryption scheme and garbled circuits. We also applied our technique to the case where the underlying OPE scheme is probabilistic. Finally, we implemented our scheme with SCAPI and an optimized Paillier's scheme and showed that it achieves acceptable performance for interactive use.


\bibliographystyle{IEEEtranS}
\bibliography{IEEEabrv,ObliviousOPE}


\appendix


\subsection{Proof of Theorem \ref{Theorem_OOPE_Security}}
\label{Theorem_OOPE_Security_App}

\textbf{Theorem \ref{Theorem_OOPE_Security}} (Security)\textbf{.}
The protocol $\Pi$\textsubscript{\tiny \normalfont OOPE} securely implements the OOPE functionality in the semi-honest model with honest majority.
\begin{proof} Since the protocol makes a call to the comparison functionality involving the DO and the DA, the proof will use the simulators of Theorem 7 of  \cite{LindellP.2009} to generate the view of the Data Owner and the Data Analyst. Let $S_{1}^{CO}$ and $S_{2}^{CO}$ be respectively the simulators of the DO and the DA in the comparison protocol. We follow the idea of \cite{LindellP.2009} by proving the cases separately, when the DO is corrupted, when the DA is corrupted and when the CSP is corrupted. Let $CO$ denote the comparison functionality of the garbled circuit protocol, in which the DO and the DA are respectively generator and evaluator, then the view of DA and DO in the comparison protocol are denoted by $view_{DA}^{\Pi_{CO}}$ and $view_{DO}^{\Pi_{CO}}$ respectively.\\

\noindent\textbf{Case 1 - DO is corrupted}\\
The view of the Data Owner consists of a sequence of randomized inputs and its view in the comparison steps needed to compute the output. Let $l$ be the number of comparisons
required to encrypt $\overline{x}$, then $view_{DO}^{\Pi\textsubscript{\tiny \normalfont OOPE}}(\mathbb{S}, \overline{x}, sk)$ contains:
\begin{align}
&[\![x^{(i)}+r^{(i)}]\!], x^{(i)}+r^{(i)}, b_o^{(i)}, b_o^{'(i)}, \nonumber \\
&view_{DO}^{\Pi_{CO}}(\langle b_o^{(i)}, b_o^{'(i)}, x^{(i)}+r^{(i)}\rangle,  \langle b_a^{(i)}, b_a^{'(i)}, \overline{x}^{(i)}+r^{(i)}\rangle), \nonumber \\
&\langle b_e^{(i)} \oplus b_a^{(i)}\oplus b_o^{(i)}, b_g^{(i)} \oplus b_a^{'(i)}\oplus b_o^{'(i)} \rangle, \langle b_e^{(i)} \oplus b_a^{(i)}, b_g^{(i)} \oplus b_a^{'(i)} \rangle
\label{Eqn_DO_VIEW}
\end{align}
for $i = 1 \cdots l$.
Notice that, in contradiction to a normal comparison protocol, neither the real input nor the real output are revealed to the DO. They are completely random to the DO. The input is randomized by the CSP and the output is randomly blinded by the DA. 
Upon input $(sk, \emptyset)$ $S_{DO}$ generates for each $i = 1 \cdots l$ the following:  
\begin{align}
&[\![x^{'(i)} ]\!], x^{'(i)}, b^{(i)}, b^{'(i)}, \nonumber \\
&S_{1}^{CO}(\langle b^{(i)}, b^{'(i)}, x^{'(i)} \rangle, \langle b_1^{(i)}\oplus b^{(i)}, b_2^{(i)}\oplus b^{'(i)} \rangle), \nonumber \\ 
&\langle b_1^{(i)}\oplus b^{(i)}, b_2^{(i)}\oplus b^{'(i)} \rangle, \langle b_1^{(i)}, b_2^{(i)} \rangle
\label{Eqn_SDO_Output}
\end{align}
where $x^{'(i)}$ is a random integer and $b^{(i)}, b^{'(i)}, b_1^{(i)}, b_2^{(i)}$ are random bits.
Clearly the outputs of Equation \ref{Eqn_DO_VIEW} and Equation \ref{Eqn_SDO_Output} are indistinguishable from each other (Equation \ref{eq:OOPE_Security_Eqn_DO}). This is because $x^{(i)}+r^{(i)}, b_o^{(i)}, b_o^{'(i)}$ are just as random as $x^{'(i)}, b^{(i)}, b^{'(i)}$ respectively. Furthermore, since $b_a^{(i)}$ and $b_a^{'(i)}$ are randomly chosen by the DA, $b_e^{(i)} \oplus b_a^{(i)}, b_g^{(i)} \oplus b_a^{'(i)}$ are also just as random as $b_1^{(i)}, b_2^{(i)}$ respectively. The security of Yao's protocol (Theorem 7 of \cite{LindellP.2009}) finishes the proof. \\

\noindent\textbf{Case 2 - DA is corrupted}\\
This case is similar to the DO's case with the only difference that the DA knows $\overline{x}$ which is the same in each protocol round. The view $view_{DA}^{\Pi\textsubscript{\tiny \normalfont OOPE}}(\mathbb{S}, \overline{x}, sk)$ contains:
\begin{align}
&r^{(i)}, \overline{x} + r^{(i)}, b_a^{(i)}, b_a^{'(i)}, \nonumber \\
&view_{DA}^{\Pi_{CO}}(\langle b_o^{(i)}, b_o^{'(i)}, x^{(i)}+r^{(i)}\rangle, \langle b_a^{(i)}, b_a^{'(i)}, \overline{x}+r^{(i)} \rangle), \nonumber \\
&\langle b_e^{(i)}\oplus b_a^{(i)} \oplus b_o^{(i)}, b_g^{(i)}\oplus b_a^{'(i)} \oplus b_o^{'(i)} \rangle, \langle b_e^{(i)} \oplus b_o^{(i)}, b_g^{(i)} \oplus b_o^{'(i)} \rangle
\label{Eqn_DA_VIEW}
\end{align}
for $i = 1 \cdots l$. 
Notice that also the DA is unaware of the result of the comparison, because the output is randomized by a bit of the DO. The simulator for the DA works in the same way as $S_{DO}$. On input $(\overline{x}, \overline{y})$ $S_{DA}$ generates for each $i = 1 \cdots l$ the following:  
\begin{align}
&r^{'(i)}, \overline{x} + r^{'(i)}, b^{(i)}, b^{'(i)}, \nonumber \\
&S_{2}^{CO}(\langle b^{(i)}, b^{'(i)}, \overline{x}+r^{(i)} \rangle, \langle b_1^{(i)}\oplus b^{(i)}, b_2^{(i)}\oplus b^{'(i)} \rangle), \nonumber \\
&\langle b_1^{(i)}\oplus b^{(i)}, b_2^{(i)}\oplus b^{'(i)} \rangle, \langle b_1^{(i)}, b_2^{(i)} \rangle
\label{Eqn_SDA_Output}
\end{align}
where $b^{(i)}, b^{'(i)}, b_1^{(i)}, b_2^{(i)}$ are random bits. \\

\noindent\textbf{Case 3 - CSP is corrupted}\\
The view $view_{CSP}^{\Pi\textsubscript{\tiny \normalfont OOPE}}(\mathbb{S}, \overline{x}, sk)$ of the CSP consists of random integers and outputs of the garbled circuit, that it receives from the DA and the DO:
\begin{align}
&[\![x^{(i)}]\!], r^{(i)}, [\![x^{(i)}+r^{(i)}]\!], \langle b_a^{(i)}, b_a^{'(i)}, b_e^{(i)}\oplus b_o^{(i)}, b_g^{(i)}\oplus b_o^{'(i)} \rangle, \nonumber \\
&\langle b_o^{(i)}, b_o^{'(i)}, b_e^{(i)}\oplus b_a^{(i)}, b_g^{(i)}\oplus b_a^{'(i)} \rangle.
\label{Eqn_CSP_View}
\end{align}
$S_{CSP}$ is given the server state $\mathbb{S}$ and a valid ciphertext $\left\langle [\![\overline{x}]\!],\overline{y} \right\rangle$. Then it chooses two elements $\left\langle [\![x_i]\!],y_j \right\rangle$,  $\left\langle [\![x_{j+1}]\!],y_{j+1} \right\rangle$ from the OPE-table, such that $y_{j} \leq \overline{y} < y_{j+1}$. The next step is to insert $[\![\overline{x}]\!]$ in the tree and simulate the path from the root to $[\![\overline{x}]\!]$. There are three possible cases:
\begin{itemize}
	\item if $y_1 = \overline{y}$ then $[\![\overline{x}]\!]$ is the same as $[\![x_j]\!]$
	\item else if $depth([\![y_j]\!]) > depth([\![y_{j+1}]\!])$ then insert $[\![\overline{x}]\!]$ right to $[\![x_j]\!]$
	\item else $depth([\![x_{j+1}]\!]) > depth([\![x_j]\!])$, insert $[\![\overline{x}]\!]$ left to $[\![x_{j+1}]\!]$
\end{itemize}
where $depth([\![x]\!])$ represents the depth of the node $[\![x]\!]$, i.e. the number of edges from the root node of the tree to $[\![x]\!]$.  
For all ancestors of $[\![\overline{x}]\!]$, $b_g^{(i)}$ is $0$ (resp. $1$) if the path $P$ goes to the left (resp. to the right). 
The value of $b_e^{(i)}$ is $1$ for all ancestors of $[\![\overline{x}]\!]$. For the node $[\![\overline{x}]\!]$ itself, there are two possible cases:
\begin{itemize}
	\item $[\![\overline{x}]\!]$ is not a leaf: this occurs if one is trying to insert a value, that was already in the tree. It holds $b_g^{(i)} = b_e^{(i)} = 0$ because $y_i$ is equal to $\overline{y}$.
	\item $[\![\overline{x}]\!]$ is a leaf: this occurs either because $[\![\overline{x}]\!]$ is inserted at a leaf node or $y_i = \overline{y}$ holds as in the first case. For the former case $b_g^{(i)}$ and $b_e^{(i)}$ are undefined because 
	no comparison was done. For the latter one $b_e^{(i)}$ is $0$ which also implies $b_g^{(i)}=0$. Hence the simulator chooses $b_g^{(i)}=b_e^{(i)}$ randomly between $0$ and undefined. 
\end{itemize}
To simulate the CSP's view, $S_{CSP}$ chooses a random integer $r^{'(i)}$ and random bits $b_{\alpha}^{(i)}, b_{\alpha}^{'(i)}$ and $b_{\omega}^{(i)}, b_{\omega}^{'(i)}$ and outputs 
\begin{align}
&[\![x^{(i)}]\!], r^{'(i)}, [\![x^{(i)}+r^{'(i)}]\!], \langle b_{\alpha}^{(i)}, b_{\alpha}^{'(i)}, b_e^{(i)}\oplus b_{\omega}^{(i)}, b_g^{(i)} \oplus b_{\omega}^{'(i)} \rangle, \nonumber \\
&\langle b_{\omega}^{(i)}, b_{\omega}^{'(i)}, b_e^{(i)}\oplus b_{\alpha}^{(i)}, b_g^{(i)}\oplus b_{\alpha}^{'(i)} \rangle.
\label{Eqn_SCSP_Output}
\end{align}
Since $[\![x^{(i)}]\!], b_e^{(i)}, b_g^{(i)}$ depend on the path they are the same in Equation \ref{Eqn_CSP_View} and Equantion \ref{Eqn_SCSP_Output} and $r^{(i)}, b_a^{(i)}, b_a^{'(i)}, b_o^{(i)}, b_o^{'(i)}$ are indistinguishable from $r^{'(i)}, b_{\alpha}^{(i)}, b_{\alpha}^{'(i)}, b_{\omega}^{(i)}, b_{\omega}^{'(i)}$.
\end{proof}

\end{document}